\documentclass[sigconf]{acmart}

\usepackage{natbib}

\usepackage{soul}
\usepackage{url}
\usepackage[utf8]{inputenc}
\usepackage[skip=2pt]{caption}
\usepackage{graphicx}
\usepackage{amsmath}
\usepackage{amsthm}
\usepackage{booktabs}
\usepackage{algorithmic}
\usepackage[switch]{lineno}

\usepackage{multirow}
\usepackage{array}      
\usepackage{makecell}   

\usepackage{subcaption}
\usepackage{makecell} 

\usepackage{tabularx}  
\usepackage{makecell}  
\usepackage{float}

\usepackage{colortbl}
\usepackage{pgf}

\usepackage[linesnumbered,ruled,vlined]{algorithm2e}
\usepackage{placeins} 

\newcommand{\heatmap}[1]{%
    \pgfmathsetmacro{\val}{#1}%
    \pgfmathsetmacro{\r}{ifthenelse(\val<50, \val/50, 1)}%
    \pgfmathsetmacro{\g}{ifthenelse(\val<50, 1, 2-\val/50)}%
    \pgfmathsetmacro{\b}{0}%
    \edef\temp{\noexpand\cellcolor[rgb]{\r,\g,\b}}%
    \temp #1%
}

\usepackage{mathtools}

\newtheorem{theorem}{Theorem}
\newtheorem{lemma}[theorem]{Lemma}

\newcommand{\cP}{\mathcal{P}}
\newcommand{\cD}{\mathcal{D}}

\newcommand{\kbar}{\bar{k}}

\newcommand{\KHC}{\texttt{RkH}\,}
\newcommand{\MHC}{\texttt{M2hC}\,}
\newcommand{\MRC}{\texttt{MRC}\,}

\newcommand{\threepointfive}{\texttt{GPT-3.5-turbo\,}}
\newcommand{\four}{\texttt{GPT-4o-mini\,}}
\newcommand{\five}{\texttt{GPT-5-mini\,}}

\newcommand{\podcast}{\texttt{podcast\,}}
\newcommand{\news}{\texttt{news\,}}
\newcommand{\semi}{\texttt{semiconductor\,}}
\newcommand{\ms}{\texttt{microsoft\,}}

\usepackage{soul}

\AtBeginDocument{%
  }

\copyrightyear{2026}
\acmYear{2026}
\setcopyright{cc}
\setcctype{by}
\acmConference[KDD '26]{Proceedings of the 32nd ACM SIGKDD Conference on Knowledge Discovery and Data Mining V.2}{August 09--13, 2026}{Jeju Island, Republic of Korea}
\acmBooktitle{Proceedings of the 32nd ACM SIGKDD Conference on Knowledge Discovery and Data Mining V.2 (KDD '26), August 09--13, 2026, Jeju Island, Republic of Korea}
\acmDOI{10.1145/3770855.3818007}
\acmISBN{979-8-4007-2259-2/2026/08}

\begin{document}

\title{Core-based Hierarchies for Efficient GraphRAG}

\author{Jakir Hossain}
\affiliation{%
  \institution{University at Buffalo}
  \city{Buffalo, NY}
  \country{USA}}
\email{mh267@buffalo.edu}

\author{Ahmet Erdem Sar{\i}y\"{u}ce}\authornote{Amazon Web Services. This publication describes work performed at University at Buffalo, and is not associated with Amazon Web Services.}
\affiliation{%
  \institution{University at Buffalo}
  \city{Buffalo, NY}
  \country{USA}}
\email{erdem@buffalo.edu}

\begin{abstract}
Retrieval-Augmented Generation (RAG) enhances large language models by incorporating external knowledge. However, existing vector-based methods often fail on global sensemaking tasks that require reasoning across many documents. 
GraphRAG addresses this by organizing documents into a knowledge graph with hierarchical communities that can be recursively summarized.
Current GraphRAG approaches rely on Leiden clustering for community detection, but we prove that on sparse knowledge graphs, where average degree is constant and most nodes have low degree, modularity optimization admits exponentially many near-optimal partitions, making Leiden-based communities inherently non-reproducible. To address this, we propose replacing Leiden with $k$-core decomposition, which yields a deterministic, density-aware hierarchy in linear time. 
We introduce a set of lightweight heuristics that leverage the $k$-core hierarchy to construct size-bounded, connectivity-preserving communities for retrieval and summarization, along with a token-budget–aware sampling strategy that reduces LLM costs. 
We evaluate our methods on real-world datasets including financial earnings transcripts, news articles, and podcasts, using three LLMs for answer generation and five independent LLM judges for head-to-head evaluation. Across datasets and models, our approach consistently improves answer comprehensiveness and diversity while reducing token usage, demonstrating that $k$-core-based GraphRAG is an effective and efficient framework for global sensemaking.
\end{abstract}



\maketitle

\newcommand\kddavailabilityurl{https://doi.org/10.5281/zenodo.20500254}
\ifdefempty{\kddavailabilityurl}{}{
\begingroup\small\noindent\raggedright\textbf{Resource Availability:}\\
The source code of this paper has been made publicly available at \url{\kddavailabilityurl}.
\endgroup
}

\section{Introduction}

Retrieval-Augmented Generation (RAG) enhances large language models (LLMs) with external knowledge, enabling more grounded and accurate responses to complex queries \cite{lewis20,baumel18,yao17,laskar20,gao23}. It is particularly useful when the corpus is too large to fit within an LLM’s context window \cite{kuratov24, liu24}. In a typical RAG setup, a small set of relevant records is retrieved and fed, along with the query, to the LLM for response generation. RAG has been widely adopted in domains such as healthcare, law, finance, and education, improving faithfulness, reducing hallucinations, mitigating privacy risks, and enhancing robustness \cite{xu24,wiratunga24,zhang23,miladi24,zhao23}.

Most RAG systems rely on lexical/semantic retrieval over text and are effective only for queries answerable from a few documents. These approaches struggle with multi-hop reasoning or queries requiring synthesis across semantically diverse documents or global corpus understanding \cite{liu24,kuratov24}. For instance, answering ``How have patient outcomes in cancer treatment evolved in response to multi-center clinical trials over the last 15 years?'' demands aggregating evidence across many sources, which conventional RAG methods cannot easily support.
Similarly, a financial analyst asking ``What common strategies have semiconductor companies adopted in response to supply chain disruptions over the past decade?'' must synthesize themes across hundreds of earnings call transcripts, a task that no single retrieved passage can answer.
We refer to such queries as \emph{global sensemaking} tasks: they require the system to identify recurring themes, reconcile conflicting perspectives, and synthesize evidence distributed across the entire corpus, capabilities that go beyond standard RAG.

Graph-based Retrieval-Augmented Generation (GraphRAG) has recently emerged as a promising approach for global sensemaking tasks by representing the corpus as a knowledge graph and organizing it into a hierarchy of communities that can be recursively summarized~\cite{han2025}. 
The GraphRAG framework introduced by \citeauthor{edge24} relies on modularity-based community detection, such as Leiden \cite{traag19}. However, Leiden often produces hierarchies that are shallow, overly fragmented, and/or dominated by a few large communities.
Building on an observation by Good, de Montjoye, and Clauset~\cite{good2010performance}, we prove that when the average degree is constant and most nodes have low degree---the regime of typical knowledge graphs---the number of near-optimal modularity partitions is exponential in the graph size (Theorem~\ref{thm:main}). This implies that Leiden-based communities are inherently non-reproducible on such graphs, compounding known issues with resolution limits and sensitivity to initialization.
In practice, this instability causes communities to unpredictably merge or fragment semantically meaningful structures.

Motivated by these limitations, we explore \emph{$k$-core–based alternatives} for hierarchical graph construction in GraphRAG. 
The $k$-core decomposition organizes a network into nested layers of increasing minimum degree~\cite{seidman83}. Each node is assigned a core number: the largest $k$ for which it belongs to a subgraph where every node has at least $k$ neighbors. 
This creates a deterministic, density-aware hierarchy that can be computed in a single $O(|E|)$ pass.
Unlike modularity-driven methods, the $k$-core hierarchy naturally captures progressively denser and more cohesive substructures, making it well suited for global sensemaking tasks that require integrating evidence across multiple interconnected regions of the corpus.
In knowledge graphs specifically, higher $k$-cores correspond to entities connected through multiple distinct relational paths, providing a natural proxy for topical centrality that modularity's comparison to a degree-preserving null model cannot capture.

We operationalize this hierarchy through a set of lightweight, interpretable community construction heuristics that operate over different levels of the $k$-core decomposition. 
Our approach constructs balanced hierarchical communities while explicitly controlling cluster sizes to respect LLM context constraints. These heuristics enable comprehensive and diverse summaries without excessive token usage.
Additionally, we introduce a token-budget–aware sampling strategy that further reduces LLM costs while preserving retrieval quality.
We make the following key contributions:

\begin{itemize}
\item We introduce $k$-core decomposition as a drop-in replacement for Leiden in GraphRAG, yielding deterministic, density-aware hierarchies in linear time.
\item We prove that modularity optimization on sparse graphs admits exponentially many near-optimal partitions (Theorem~\ref{thm:main}), formally explaining why Leiden-based community detection is unreliable on knowledge graphs.
\item We propose multiple hierarchical heuristic strategies that exploit different regions of the induced $k$-core hierarchy to balance coverage, granularity, and efficiency.
\item We conduct extensive evaluations on three real-world datasets, using three LLMs for answer generation and five independent LLM judges for head-to-head evaluation.
\end{itemize}

\section{Related Work and Background}\label{sec:back}

In this section, we review the literature on Graph-based RAG and remind key concepts of GraphRAG workflows and $k$-core decomposition.

\noindent {\bf Graph-based RAG.} Traditional vector-based RAG retrieves relevant passages to augment LLM responses, but is largely limited to single-hop, fact-based queries where the answer resides within one or a few contiguous passages. When reasoning must span multiple documents or synthesize information across a corpus, these approaches fall short \cite{han2025, peng25}.

To address this, several methods incorporate graph structures into the retrieval process. One line of work constructs query-focused subgraphs from knowledge graphs and performs reasoning over them. Systems such as KG-GPT \cite{kim23}, G-Retriever \cite{he24}, GRAG \cite{hu25}, and HLG~\cite{Ghassel25} follow this paradigm, extracting relevant subgraphs and using LLMs to reason over the retrieved structure, thereby improving multi-hop answer quality. A complementary line of work enhances the retrieval step itself with graph signals, using graph topology to guide which documents or passages are selected \cite{sun23,wang24,peng25}.
However, both subgraph-reasoning and graph-enhanced retrieval methods are primarily designed for queries that require a bounded number of reasoning hops. They do not address \emph{global sensemaking}, tasks where the answer requires understanding themes, patterns, or relationships that emerge only when considering an entire corpus.

GraphRAG \cite{edge24} addresses this gap by applying Leiden-based hierarchical community detection on a knowledge graph and using the resulting communities for global sensemaking (see Section~\ref{sec:graphrag_overview} for details).
A key limitation, however, is that modularity-based clustering methods like Leiden can produce highly uneven hierarchies, where community sizes vary widely across levels. This imbalance can lead to inconsistent summary quality and makes it difficult to control the granularity of retrieved context. Our work addresses this limitation by replacing community detection with $k$-core decomposition, which yields a naturally nested and more balanced hierarchy of dense subgraphs.

\subsection{Community-based GraphRAG Overview} \label{sec:graphrag_overview}

\citeauthor{edge24} \cite{edge24} introduced a graph-based approach for global sensemaking question answering by a community-driven GraphRAG workflow. Here, we briefly describe the two main stages of the GraphRAG workflow: indexing and query-time answer generation, followed by their evaluation conditions.

\noindent {\bf Indexing.} Large text corpora are first split into manageable chunks, from which entities, relationships, and claims are extracted using LLMs. These elements are aggregated into a knowledge graph, where nodes represent entities or claims and edges represent relationships, weighted by frequency. The graph is hierarchically partitioned into communities using Leiden community detection \cite{traag19}, recursively identifying sub-communities until reaching leaf-level communities that cannot be further split.
Leaf-level communities are summarized first by prioritizing elements according to edge prominence, adding source and target node descriptions, edges, and related claims into the LLM context until the token limit is reached. For higher-level communities, if all element summaries fit within the context window, they are summarized directly; otherwise, sub-community summaries iteratively replace longer element-level descriptions until they fit. This process, from knowledge graph construction to hierarchical community summaries, is referred to as indexing.

\noindent {\bf Query-Time Answer Generation.}
Given a user query, indexed community summaries are used to generate a final answer through a multi-stage Map-Reduce process. First, summaries are randomly shuffled and divided into chunks of a pre-specified token size to ensure relevant information is distributed across multiple contexts rather than concentrated in a single window. Next, intermediate answers are generated for each chunk in parallel, with the LLM also assigning a helpfulness score between 0 and 100. These intermediate answers are then sorted by helpfulness and iteratively added into a new context window until the token limit is reached. The LLM uses this aggregated context to produce the final global answer. By leveraging the hierarchical structure of the communities, this approach allows answers to be generated from different levels and enables robust aggregation of information from all relevant sub-communities.

\noindent {\bf Evaluation for GraphRAG.} Evaluating RAG systems for multi-document reasoning presents distinct challenges beyond standard QA benchmarks. Recent work assesses output quality along dimensions such as comprehensiveness, diversity, empowerment, and directness \cite{edge24}. 
\citeauthor{edge24} tested six conditions, including four hierarchical community levels (C0--C3), source-text summarization, and a vector-based RAG approach. The hierarchical configurations differ in granularity. \textbf{C0} uses only top-level communities, yielding the smallest set of summaries, \textbf{C1} uses communities one level below the root for a slightly finer decomposition, and so on. Results on two datasets show that all community-level approaches (C0--C3) outperform both source-text summarization and vector RAG for global sensemaking. Among these, \textbf{C2} and \textbf{C3} consistently achieve the best performance, hence we consider those two for our evaluation.

\subsection{Hierarchical $k$-core Decomposition}

In this work, we consider the knowledge graph as an unweighted graph $G = (V, E)$, where $V$ denotes the set of nodes and $E$ denotes the set of edges.
The degree of a node $u$ within a subgraph $S \subseteq G$ is denoted $deg(u, S)$.
A $k$-core is the maximal connected subgraph $S \subseteq G$ in which every vertex has at least $k$ neighbors, i.e., $deg(u, S) \ge k~\forall u \in S$. Each node is assigned a \textit{core number}, the largest $k$ for which it belongs to a $k$-core. The $k$-shell is the set of nodes with core number $k$~\cite{carmi07}. The $k$-cores for all $k$ can be computed by recursively removing nodes with degree less than $k$ and their edges, which runs in $O(|E|)$ time~\cite{batagelj11}. 
This process of identifying all $k$-cores and assigning core numbers is referred to as core decomposition.
Since nodes can belong to multiple $k$-cores with different $k$ values, the decomposition naturally forms a hierarchy of nested dense subgraphs.  
This hierarchy can be represented as a tree, with nodes denoting subgraphs and edges encoding containment relationships.
The root represents the entire graph (1-core), while child nodes correspond to subgraphs with increasing core numbers.
Internal nodes may contain denser subgraphs with higher core values, forming a multi-level structure that captures the nested organization of densely connected regions within the graph.

\section{Why Modularity Optimization is Unreliable on Sparse Knowledge Graphs?}

Modularity-based community detection methods such as Leiden~\cite{traag19} compare observed within-community edge density to the expectation under a degree-preserving null model (the configuration model), an approach that works well in dense networks but degrades in sparse ones. Good, de Montjoye, and Clauset~\cite{good2010performance} provided a systematic characterization of this degradation. They showed that the modularity function ${Q}$ typically admits an exponentially large number of distinct high-scoring partitions while lacking a clear global maximum: a phenomenon they termed the \emph{degeneracy of modularity}. Using real-world metabolic networks, they demonstrated that these near-optimal partitions can fundamentally disagree on key properties such as the composition of the largest modules and the distribution of module sizes, meaning that any optimizer (including Louvain and Leiden) is effectively selecting one solution from a vast equivalence class, guided more by random seeds and tie-breaking rules than by genuine structure. They further showed that $Q_{max}$
 depends strongly on both network size and the number of modules, and that degeneracy is most severe in sparse and hierarchically structured networks.

{\bf Modularity Degeneracy.} We give a key structural observation: the low-degree nodes (degree at most a fixed constant $d$) have vanishing modularity sensitivity in sparse graphs, because their few edges contribute negligibly to the objective. Denote by $n_{\mathrm{\le d}} = |\{i \in V : k_i \leq d\}|$ the number of such nodes. For small $d$, these nodes can be reassigned across communities with negligible impact on modularity, creating an exponentially large degenerate set.

Recall that the modularity of a partition $\sigma$ of a graph $G = (V,E)$ with $n$ nodes and $m$ edges is $Q(\sigma) = \frac{1}{2m}\sum_{i,j}[A_{ij} - k_ik_j/2m]\,\delta(\sigma_i,\sigma_j)$, where $k_i$ is the degree of node $i$. Let $Q^* = \max_\sigma Q(\sigma)$ denote the optimal modularity. 
For $\varepsilon > 0$, we define the \emph{$\varepsilon$-degeneracy} of $G$ as the number of structurally distinct partitions achieving near-optimal modularity: $\cD(\varepsilon) = |\{\sigma \in \cP(G) : Q^* - Q(\sigma) < \varepsilon\}|$.

\begin{theorem}[Modularity Degeneracy in Sparse Graphs]\label{thm:main}
Let $G$ be a graph with $n$ nodes, $m$ edges, average degree $\kbar = 2m/n = O(1)$, and $n_{\mathrm{\le d}} = \Theta(n)$. Then for any $\varepsilon > d(2+\kbar)/(2m)$,
\[
\cD(\varepsilon) \;\geq\; 2^{\,n_{\mathrm{\le d}}/(d+1)}.
\]
In particular, the number of near-optimal partitions is exponential in $n$, and the tolerance threshold $\varepsilon$ required to trigger this blowup is $O(1/n)$.
\end{theorem}

Proof sketch, the full proof, and experimental validation are given in Appendix~\ref{app:proof}.

Knowledge graphs (KG) are particularly well-suited to $k$-core decomposition over modularity-based methods for two structural reasons. First, their heavy-tailed but low-mean degree distributions mean that $n_{\leq d}=\Theta(n)$, placing them squarely in the regime where Theorem~\ref{thm:main} bites hardest. Second, KG edges are semantically meaningful rather than stochastic, so the $k$-core criterion, requiring each node to have at least
$k$ neighbors, directly captures genuine relational richness in a way that comparison to a degree-preserving null model does not.

\section{A Robust Alternative: $k$-core Decomposition}\label{sec:proposed_heuristics}

Theorem~\ref{thm:main} exposes a fundamental instability: for a graph with low average degree, any modularity optimizer, Leiden included, selects one partition from an exponentially large near-optimal plateau. Different seeds, tie-breaking rules, or minor edge perturbations yield structurally different communities.

The $k$-core decomposition sidesteps this problem entirely. It is \emph{unique}: for every graph $G$ and integer $k$, the $k$-core $H_k$ is the unique maximal subgraph with minimum degree $\geq k$, computed by the deterministic peeling process; no optimization landscape, no stochasticity. It is \emph{robust}: adding or removing a few edges changes the shells by at most a proportional number of nodes, whereas Theorem~\ref{thm:main} implies that even a single edge can shift the optimal partition across the entire degenerate set. And $k$-core decomposition provides a \emph{natural hierarchy}: the nested shells $H_1 \supseteq H_2 \supseteq \cdots$ reflect structural connectivity rather than comparison to a null model. In the sparse regime where $k_i \cdot k_j/2m \approx 0$ for most pairs, the null model is nearly vacuous yet $k$-core structure remains informative. The 2-core captures the backbone of multiply-connected nodes while degree-1 nodes form the periphery. In knowledge graphs, this distinction is semantically meaningful: multiple relational paths between entities signal genuine topical relatedness. Last, but not least, computing $k$-core is much cheaper than Leiden; a single pass over the graph suffices to compute all and a hierarchy~\cite{Sariyuce16}.

In GraphRAG, community assignments determine which entities are summarized together and which retrieval units are constructed. Exponential degeneracy means Leiden-based pipelines produce non-reproducible summaries: the same knowledge graph under different random seeds yields different communities and different retrieval behavior, with communities that are either too fragmented (splitting apart a natural topic) or too arbitrary (grouping unrelated peripheral nodes together because the optimization found a marginal modularity gain). The
$k$-core decomposition eliminates this variance. 
Its nested shells provide a small number of progressively tighter subgraphs where the innermost cores represent the most interconnected, and therefore most semantically central concepts, while the outer shells supply context. This maps naturally onto a summarization hierarchy: summarize the dense core first, then expand outward, producing stable, deterministic retrieval units from which summaries and indices can be reliably built.

\noindent{\bf Remark.} A natural refinement of $k$-core is the $k$-truss~\cite{cohen08}, which requires every edge to participate in at least $k-2$ triangles, yielding tighter subgraphs. However, knowledge graphs are triangle-poor: edges typically connect entities through distinct relation types (e.g., \textit{born\_in}, \textit{capital\_of}), producing bipartite-like local structure with global clustering coefficients well below 0.05. Even the 3-truss discards the vast majority of the graph, making the degree-based $k$-core criterion the right granularity for sparse relational graphs. A brief evaluation is provided in Appendix~\ref{sec:altcomdet}.

Motivated by these advantages, we propose a set of heuristics for efficient global sensemaking in GraphRAG. These heuristics leverage the $k$-core hierarchy to generate size-constrained clusters, merge small two-hop clusters, and handle residual components.

\begin{figure}[!t]
    \centering
    \includegraphics[width=0.48\textwidth]{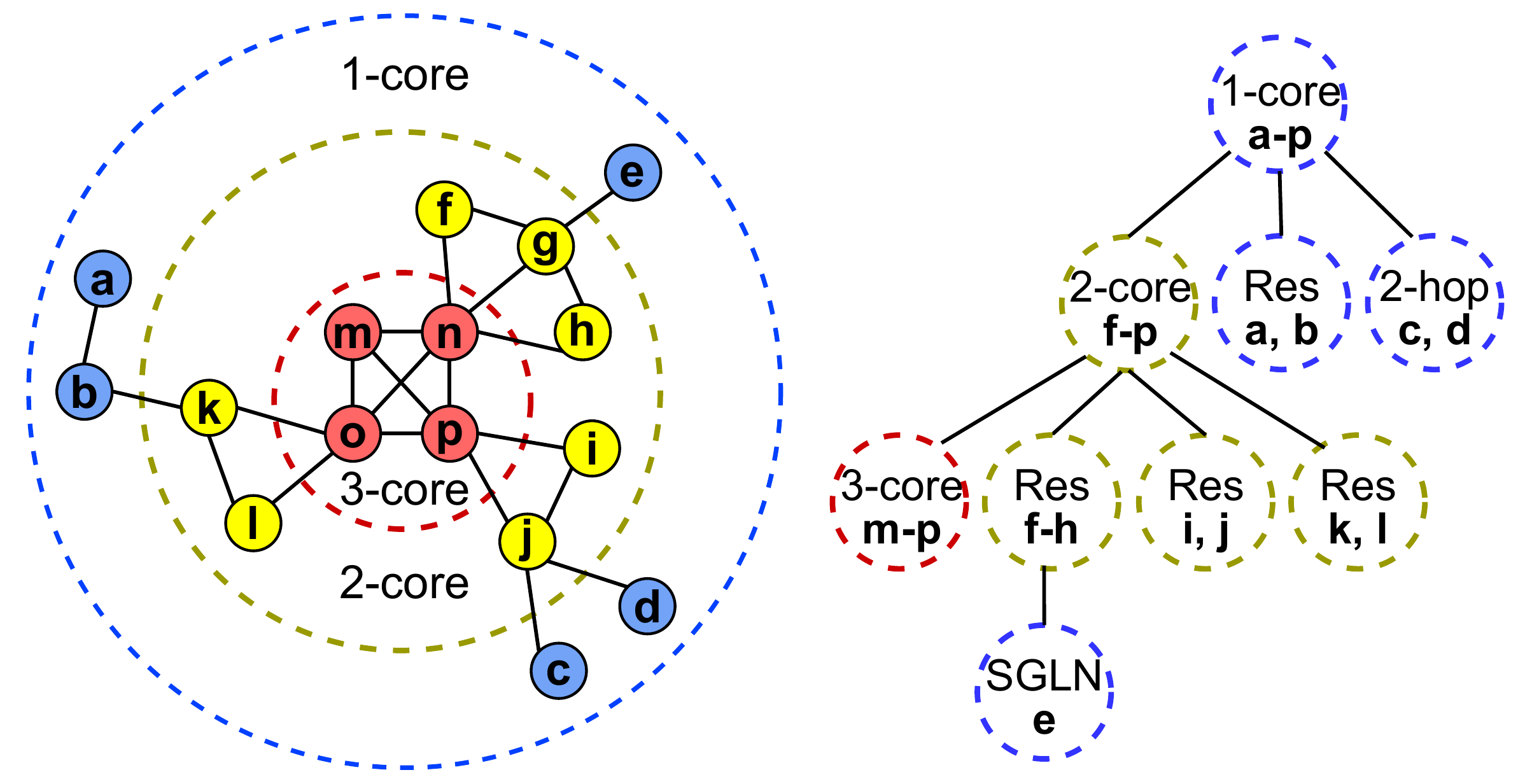} 
    \caption{$k$-core decomposition (left) and corresponding hierarchy tree produced by \KHC (right). 
}
    \label{fig:RkH_example}
\end{figure}

\subsection{Handling Residuals in $k$-core Hierarchy }

Our first heuristic builds on the $k$-core hierarchy while addressing practical issues such as singleton nodes and oversized clusters by separating dense cores from sparse residuals.
At each level, dense components are recursively partitioned into size-bounded, connectivity-preserving clusters, while low-core residuals are handled separately to avoid distorting core structure.
The result is a deterministic hierarchy that captures both central regions and peripheral context, enabling reliable hierarchical summarization.

Algorithm~\ref{alg:split_large} presents the resulting {\bf Residual-aware $k$-core Hierarchy} heuristic, \KHC in short.
It takes two parameters: the input graph ($G$), and the maximum cluster size ($M$).
\KHC begins by extracting the largest connected component and removing self-loops (line~\ref{line:lcc}), as done in GraphRAG~\cite{edge24}. It then computes the core number for each node (line~\ref{line:core}). 
The hierarchy construction is initialized by setting $\ell=1$, enqueuing full node set $V$ into $\mathcal{Q}$, and initializing an empty cluster set $\mathcal{C}$ and global singleton set $\mathcal{G}_s$ (lines~\ref{line:init_queue}–\ref{line:init_cluster}).

The graph is then processed iteratively over increasing core levels $\ell$ until the maximum core is reached (line~\ref{line:while_start}), splitting each cluster $S$ into core ($c(v)\ge\ell$) and residual ($c(v)<\ell$) nodes (line~\ref{line:sc_sr}).
Core components smaller than or equal to $M$ are added directly to $\mathcal{C}$ and the queue; larger components are split into size-bounded clusters via \textsc{Split} before adding (lines~\ref{line:for_KC}–\ref{line:add_KC}).
The maximum cluster size ($M$) is set by dividing the input token limit by the average tokens per node, providing a practical estimate that keeps most clusters within the allowed input size.
Residual nodes are handled similarly, either added directly or \textsc{Split}  as needed (lines~\ref{line:for_R}–\ref{line:add_R}); since these residual subgraphs correspond to leaf clusters, they are not pushed to the next-level queue.
The \textsc{Split} procedure breaks each oversized cluster into smaller, size-bounded clusters by greedily growing each cluster from a high-degree seed node, prioritizing neighbors that maximize internal connectivity. This ensures that dense regions remain intact while respecting the maximum cluster size; full algorithmic details are provided in Algorithm~\ref{alg:split_cc} in Appendix~\ref{sec:split_cc}.

\begin{algorithm}[!t]
\small
\caption{\KHC: Residual-aware $k$-core Hierarchy\,($G$,\,$M$)}
\label{alg:split_large}
\begin{algorithmic}[1]
\REQUIRE Graph $G=(V,E)$, max cluster size $M$ \label{line:req}
\ENSURE Hierarchical clusters $\mathcal{C}$ \label{line:ensure}

\STATE $G \leftarrow$ LCC$(G)$; remove self-loops \label{line:lcc}
\STATE Compute core number $c(v)$ for all $v \in V$ \label{line:core}
\STATE Initialize queue $\mathcal{Q} \leftarrow \{(V)\}$ and level $\ell \leftarrow 1$ \label{line:init_queue}
\STATE Initialize cluster set $\mathcal{C} \leftarrow \emptyset$, Global singleton $\mathcal{G}_s \leftarrow \emptyset $ \label{line:init_cluster}
\WHILE{$\ell \le \max c(v)$} \label{line:while_start}
  \STATE $\mathcal{Q}' \leftarrow \emptyset$ \label{line:qprime}
  \FOR{each cluster $S \in \mathcal{Q}$} \label{line:for_S}
    \STATE $S_c \leftarrow \{v\in S \mid c(v)\ge \ell\}$, $S_r \leftarrow S\setminus S_c$ \label{line:sc_sr}
    \FOR{each connected component $R$ of $G[S_c]$} \label{line:for_KC}
          \STATE $R' \gets
          \begin{cases}
      \{R\}, & |R| \le M \\
      \textsc{Split}(G, R, M) & \text{otherwise}
      \end{cases}$
        \STATE $\mathcal{C} \gets \mathcal{C} \cup R'$,  $\mathcal{Q}' \gets \mathcal{Q}' \cup R'$   \label{line:add_KC}
    \ENDFOR
    \FOR{each connected component $R$ of $G[S_r]$} \label{line:for_R}
      \STATE $\mathcal{C} \gets \mathcal{C} \cup 
      \begin{cases}
      \{R\}, & |R| \le M \\
      \textsc{Split}(G, R, M) & \text{otherwise}
      \end{cases}$ \label{line:add_R}
    \ENDFOR
           \STATE $S_{\text{single}} \gets \{ R \in \mathcal{C} \mid |R| = 1 \}$  \label{line:collect_singletons}
    \STATE $\mathcal{H} \gets \{ R \subseteq S_{\text{single}} \mid R \text{ is 2-hop connected in } G \}$ \label{line:form_2hop}
    \STATE $\mathcal{G}_s \gets S_{\text{single}} \setminus \bigcup_{R \in \mathcal{H}} R$ \label{line:collect_global_singletons}

      \label{line:group_2hop}
    \FOR{each 2-hop component $R$ of $\mathcal{H}$} \label{line:for_H}
      \STATE $\mathcal{C} \gets \mathcal{C} \cup 
      \begin{cases}
      \{R\}, & |R|\le M \\
      \textsc{Split-2hop}(G, R, M) & \text{otherwise}
      \end{cases}$ \label{line:add_H}
    \ENDFOR
  \ENDFOR
  \STATE $\mathcal{Q} \leftarrow \mathcal{Q}'$, $\ell \leftarrow \ell+1$ \label{line:update_queue}
\ENDWHILE \label{line:while_end}
\STATE Attach each of $\mathcal{G}_s$  to neighboring clusters in $\mathcal{C}$ \label{line:attach_singletons}

\RETURN $\mathcal{C}$ \label{line:return}
\end{algorithmic}
\end{algorithm}

Now, we extract singleton clusters from the existing set $\mathcal{C}$ (line~\ref{line:collect_singletons}) into $S_{\text{single}}$ and form new clusters $\mathcal{H}$ consisting of nodes that are 2-hop connected in $G$ (line~\ref{line:form_2hop}). 
A set $R$ is 2-hop connected if each node $u \in R$ has a path of length $\le 2$ in $G$ to some other node $v \in R$.
Any remaining singleton nodes not included in $\mathcal{H}$ are added to global set $\mathcal{G}_s$ (line~\ref{line:collect_global_singletons}).
Now, each 2-hop component from $\mathcal{H}$ is either added directly to $\mathcal{C}$ or split into smaller clusters via  \textsc{SPLIT-2HOP} if it exceeds size $M$ (lines~\ref{line:for_H}–\ref{line:add_H}). 

The \textsc{SPLIT-2HOP} procedure greedily splits a 2-hop connected component, ensuring that nodes sharing common anchors are grouped together.
Starting from a seed node with the largest anchor set, the algorithm greedily grows a cluster by adding nodes that share anchors with the current cluster, continuing until the cluster reaches the maximum size or no eligible nodes remain. 
Anchor nodes connected to multiple cluster members are then included to maintain structural coherence. 
The process produces size-bounded clusters that respect both 2-hop connectivity and anchor relationships. Full details are provided in Algorithm~\ref{alg:split2hop} in Appendix~\ref{sec:split2hop}.

After all clusters at the current level are processed, the queue is updated with the newly formed clusters, and the level counter is incremented (line~\ref{line:update_queue}). 
Once all core levels have been processed, a final attachment step assigns each node in $\mathcal{G}_s$ to a neighboring cluster in $\mathcal{C}$, ensuring that no isolated nodes remain in the hierarchy (line~\ref{line:attach_singletons}).
Finally, the algorithm returns the hierarchical cluster set $\mathcal{C}$ (line~\ref{line:return}). This approach ensures size-bounded, connectivity-preserving clusters while capturing both dense core structures and sparse residual nodes.

Figure~\ref{fig:RkH_example} shows an example hierarchy produced by \KHC.
The 1-core contains all nodes (a–p). 
At the next level, the 2-core consists of nodes (f–p), while the remaining nodes (a–e) form different clusters, including a connected component (a–b), a 2-hop connected group (c–d), and a singleton (e). 
At the 3-core level, only nodes (m–p) remain, and the residual nodes (f–l) are partitioned into three connected components (f–h), (i–j), and (k–l).
Note that the singleton (e) is attached to its neighboring cluster (f–h) at the end.
This recursive decomposition yields a hierarchy where higher $k$-cores capture denser subgraphs and residuals are organized into progressively finer-grained structures, as formalized in Algorithm~\ref{alg:split_large}.

\subsection{Handling Small Clusters}

Due to the sparse nature of knowledge graphs, many tiny clusters (often containing only two nodes) can arise from previous heuristic.
 As a result, when GraphRAG evaluates community relevance during query answering, these small clusters usually receive low scores. Consequently, they are often excluded from the final answer generation, which can affect overall sensemaking performance.
We propose to explicitly merge such small clusters.\\

\noindent {\bf Merging 2-hop Small Clusters.} In \KHC, some 2-hop clusters generated in lines~\ref{line:for_H}–\ref{line:add_H} may contain only two nodes. These very small clusters can fragment the hierarchy and reduce connectivity if left unprocessed. To address this, we apply a post-processing step after Algorithm~\ref{alg:split_large} that merges small 2-hop clusters into larger clusters whenever possible, or creates new clusters when no suitable neighbors exist.
The merging procedure is detailed in Algorithm~\ref{alg:merge_small_cluster}, referred to as \MHC.
First, all 2-hop clusters of size two are separated from the remaining clusters (lines~\ref{line:small_clusters}–\ref{line:large_clusters}).
 While small clusters remain (line~\ref{line:while_small}), the algorithm selects the small cluster with the most connections to existing clusters (line~\ref{line:pick_s}) and counts its neighbors in each cluster (line~\ref{line:count_neighbors}). 
If neighbors exist, the small cluster is merged in-place into the cluster with which it shares the most edges (lines~\ref{line:if_neighbors}–\ref{line:merge_cluster}).
Otherwise, a new cluster is created for the small cluster (line~\ref{line:no_neighbors}). After each iteration, the processed small cluster is removed from the pool (line~\ref{line:remove_s}). Finally, the algorithm returns the updated cluster set with all small clusters integrated (line~\ref{line:return_L}).\\

\begin{algorithm}[!t]
\caption{\MHC: Merge 2-hop Clusters ($G$,  $\mathcal{C}$)}
\label{alg:merge_small_cluster}
\begin{algorithmic}[1]
\REQUIRE Graph $G=(V,E)$, existing cluster set $\mathcal{C}$
\ENSURE Updated cluster set $\mathcal{C}$ with small clusters merged

\STATE $\mathcal{S} \gets$ All 2-hop clusters of size 2 in $\mathcal{C}$ \label{line:small_clusters}
\STATE $\mathcal{L} \gets \mathcal{C} \setminus \mathcal{S}$ \label{line:large_clusters}

\WHILE{$\mathcal{S} \neq \emptyset$} \label{line:while_small}
    \STATE Pick small cluster $C_s \in \mathcal{S}$ with the most neighbors in $\mathcal{L}$ \label{line:pick_s}
    \STATE Count neighbors of $C_s$ in each cluster of $L \in \mathcal{L}$ \label{line:count_neighbors}
\IF{any neighbors exist} \label{line:if_neighbors}
    \STATE Find $L_{best} \in \mathcal{L}$ with most neighbors of $C_s$  
     \STATE $L_{best} \gets L_{best} \cup C_s$ \label{line:merge_cluster}
\ELSE \label{line:no_neighbors}
    \STATE $\mathcal{L} \gets \mathcal{L} \cup C_s$ \label{line:new_cluster}
\ENDIF

    \STATE Remove $C_s$ from $\mathcal{S}$ \label{line:remove_s}
\ENDWHILE

\RETURN $\mathcal{L}$  \label{line:return_L}
\end{algorithmic}
\end{algorithm}

\noindent {\bf Merging Residual Clusters.} In addition to the previous heuristics, we observe that many residual connected components produced during \KHC (lines~\ref{line:for_R}–\ref{line:add_R}) have size two. To address this, we introduce an extended heuristic, referred to as \MRC.
\MRC is conceptually similar to \MHC, but instead of focusing solely on 2-hop clusters, it handles all small residual clusters of size two. The only modification occurs in line~\ref{line:small_clusters} of Algorithm~\ref{alg:merge_small_cluster}, where both residual clusters and 2-hop clusters of size two are collected into the set $\mathcal{S}$ for processing. The rest of the procedure remains unchanged: each small cluster is merged into a neighboring parent cluster when possible; otherwise, a new cluster is created only if no suitable neighbors exist.

Note that, for \MHC and \MRC, we focus specifically on clusters of size two, as they account for the majority of small clusters observed in practice. Clusters of size three or larger occur infrequently, and trying to merge would not yield measurable improvements; moreover, merging them would unnecessarily inflate neighboring clusters and weaken size constraints.

\subsection{Token Efficiency via Sampling}\label{subsec:alg-token}

In large hierarchical graph decompositions, leaf-level communities can be densely connected, and many nodes/edges within the same community frequently carry overlapping information. Passing all of this content to an LLM for retrieval or summarization can therefore be costly in terms of tokens, leading to inefficiency and redundancy. To address this challenge, we propose the \textit{Round-Robin Token-Constrained Selection} (RRTC) heuristic, which efficiently reduces token usage while preserving the most informative edges from each community.
RRTC operates on leaf-level communities produced by any hierarchy heuristic (e.g., \KHC, \MHC, or \MRC) and selects a representative subset of edges under a fixed token budget. Within each community, edges are ranked by the combined degree of their endpoints to reflect overall prominence. Selection then proceeds in a round-robin fashion across leaf-level communities, traversing from higher to lower $k$-shells, until the token budget is exhausted. By sampling a representative subset of edges from each community, RRTC captures the essential information without including all nodes/edges.
Our design is motivated by the $k$-shell-based influential node identification methods \cite{wang20,hossain23}, where round-robin selection across $k$-shells identifies diverse key nodes.

\begin{table}[t]
\caption{Summary statistics of datasets used in our experiments, showing model, dataset, number of documents, and graph information ($|V|$: nodes, $|E|$: edges, $K_m$: maximum core number). Post-cutoff data is used for \threepointfive (cutoff: Sep 1, 2021) and \four (cutoff: Oct 1, 2023), while full data is used for \five.
}
\label{tab:datasets_stats} 
\resizebox{0.45\textwidth}{!}{%
\centering
\begin{tabular}{c l r r r r}
\toprule
\makecell{\textbf{Model}} &
\makecell{\textbf{Dataset}} &
\makecell{\textbf{\# Docs}} &
\makecell{\textbf{|V|}} &
\makecell{\textbf{|E|}} &
\makecell{$\mathbf{K}_m$} \\
\hline
\multirow{3}{*}{\makecell{\threepointfive \\ (post-cutoff)}}
 & \podcast & 36  & 1791  & 2651  & 6 \\
                 & \news    & 609 & 10655 & 15368 & 6 \\
                 & \semi    & 183 & 4324  & 9546  & 13 \\
\hline
\multirow{3}{*}{\makecell{\four \\ (post-cutoff)}}          & \podcast & 13  & 628   & 882   & 5 \\
                 & \news    & 331 & 5312  & 7316  & 6 \\
                 & \semi    & 63  & 2029  & 4249  & 9 \\
\hline
\multirow{3}{*}{\makecell{\five \\ (full)}}          & \podcast & 72  & 3323  & 7761  & 7 \\
                 & \news    & 609 & 4499 & 7941  & 8 \\
                 & \ms      & 45  & 790   & 1702  & 7 \\
\hline
\end{tabular}
}
\end{table}

\section{Experimental Setup}

In this section, we describe the experimental setup, including the datasets and evaluation approach.
Since our study targets the global sensemaking task, we selected three datasets requiring reasoning over multiple documents. The first two align with the original study of \citeauthor{edge24}, while the third also requires multi-document reasoning. Dataset sizes range from approximately 1M to 6M tokens, all processed using the same knowledge graph construction pipeline used by \cite{edge24}. Details of these datasets are reported in Table~\ref{tab:datasets_stats}.

 \noindent \textbf{Podcast transcripts (\podcast).} Public transcripts of \textit{Behind the Tech with Kevin Scott}, a podcast featuring conversations between Microsoft CTO Kevin Scott and various thought leaders in science and technology \cite{Scott24}. The corpus contains 72 documents with approximately 1 million tokens in total.
     
\noindent\textbf{News articles (\news).} A benchmark dataset of news articles spanning multiple categories, including entertainment, business, sports, technology, health, and science \cite{tang24}. 
  It is a multi-document QA dataset containing 609 documents and nearly 1.4 million tokens.
  
 \noindent\textbf{S\&P 500 earnings transcripts.} This dataset consists of earnings call transcripts from publicly listed S\&P 500 companies spanning 2013–2025 \cite{ca25}, with over 20,000 documents totaling 232M tokens.
For \threepointfive and \four, to focus on multi-document synthesis and keep the dataset manageable, we use only post-cutoff transcripts from the Semiconductor industry ({\bf \semi}), yielding 183 and 63 documents, respectively. For \five, since the full \semi dataset is large (nearly 6M tokens), we use Microsoft earnings calls from 2013–2024 ({\bf \ms}) to limit corpus size.

\subsection{Evaluation Criteria}
Following the methodology of \cite{edge24}, we generated 125 sensemaking questions using \five (details are in Appendix \ref{sec:question_generation}).  
Consistent with prior work~\cite{edge24,tiwari25,kim25}, we evaluate answers using two primary criteria: \textbf{Comprehensiveness}, which measures how thoroughly an answer addresses all relevant aspects of a question, and \textbf{Diversity}, which measures the extent to which an answer captures varied perspectives and insights.  
\citeauthor{edge24}~\cite{edge24} also report \textit{Empowerment} (ability to support informed judgment) and \textit{Directness} (degree to which an answer directly addresses the question); results for these metrics are provided in the 
Appendix (Table \ref{tab:gpt35_empowerment_directness}).  

\textbf{Retrieval approaches and evaluation.}   \label{sec:evaluation_method}
We focus on \textbf{C2} and \textbf{C3} from Section~\ref{sec:back}, which perform best overall, and evaluate all $k$-core heuristics—\KHC, \MHC, and \MRC—against them. 
Higher $k$-core levels typically yield finer-grained communities, improving answer quality via more detailed summaries (
see Appendix, Table~\ref{tab:gpt35_msft_mhc_h2h}). Hence, for all heuristics, we consider the Leaf and immediate parent levels, labeled LF and L1 (e.g., \KHC LF).  
For evaluation, we use a head-to-head framework where LLMs compare answers to select a winner, loser, or tie. To improve reliability, we use multiple evaluators and perform repeated assessments, with majority voting determining the final outcome—suitable for global sensemaking tasks without gold-standard references (see Appendix~\ref{sec:evaluation_details} for details).

\textbf{Configuration.}
For fair comparison with \cite{edge24}, we performed indexing using a 600-token chunking window with 100-token overlap and an 8k-token context window for community summaries, following their approach.
We utilized the public Azure OpenAI endpoint and the OpenAI API for \threepointfive, \four, and \five. For the multiple-judge evaluation, \five was accessed via Azure OpenAI, while the other four LLMs (Gemini 3 Pro Preview, Gemini 2.5 Pro, Qwen3 Next 80B, and DeepSeek v3.2) were accessed via GCP.  
All LLM-based evaluations were performed on a GCP virtual machine equipped with 32 GB memory and 16 vCPUs (Intel Haswell architecture).
Detailed prompts for graph construction, community summaries, and global answer generation (following \cite{edge24}), along with our code, are available here: \textbf{\href{https://github.com/erdemUB/KDD26}{https://github.com/erdemUB/KDD26}}.

\begin{table*}[!h]
\centering
\caption{\threepointfive post-cutoff: Head-to-head win rates (\%), for comprehensiveness and diversity metrics. C2 and C3 are Leiden community levels from Edge et al.~\cite{edge24}. LF is leaf-level communities, and L1 is the level immediately above the leaf.
 }
\label{tab:gpt35}
\resizebox{0.9\textwidth}{!}{%
\footnotesize
\begin{tabular}{|r|ccc|ccc|ccc|ccc|ccc|ccc|}
\hline
\threepointfive
& \multicolumn{6}{c|}{\podcast} 
& \multicolumn{6}{c|}{\news} 
& \multicolumn{6}{c|}{\semi} \\
\cline{2-19}
results for
& \multicolumn{3}{c|}{C2} & \multicolumn{3}{c|}{C3}
& \multicolumn{3}{c|}{C2} & \multicolumn{3}{c|}{C3}
& \multicolumn{3}{c|}{C2} & \multicolumn{3}{c|}{C3} \\
\cline{2-19}
{\bf Comprehensiveness}
& Win & Loss & Tie & Win & Loss & Tie
& Win & Loss & Tie & Win & Loss & Tie
& Win & Loss & Tie & Win & Loss & Tie \\
\hline
\KHC L1   
& \heatmap{56} & \heatmap{36} & \heatmap{8} & \heatmap{50} & \heatmap{42} & \heatmap{8}
& \heatmap{44} & \heatmap{44} & \heatmap{12} & \heatmap{52} & \heatmap{38} & \heatmap{10}
& \heatmap{48} & \heatmap{48} & \heatmap{4} & \heatmap{56} & \heatmap{42} & \heatmap{2} \\

\KHC LF   
& \heatmap{58} & \heatmap{42} & \heatmap{0} & \heatmap{44} & \heatmap{52} & \heatmap{4}
& \heatmap{52} & \heatmap{42} & \heatmap{6} & \heatmap{49} & \heatmap{47} & \heatmap{4}
& \heatmap{50} & \heatmap{50} & \heatmap{0} & \heatmap{50} & \heatmap{48} & \heatmap{2} \\

\MHC L1  
& \heatmap{58} & \heatmap{36} & \heatmap{6} & \heatmap{52} & \heatmap{42} & \heatmap{6}
& \heatmap{56} & \heatmap{34} & \heatmap{10} & \heatmap{50} & \heatmap{36} & \heatmap{14}
& \heatmap{58} & \heatmap{42} & \heatmap{0} & \heatmap{48} & \heatmap{42} & \heatmap{10} \\

\MHC LF  
& \heatmap{58} & \heatmap{40} & \heatmap{2} & \heatmap{52} & \heatmap{44} & \heatmap{4}
& \heatmap{54} & \heatmap{38} & \heatmap{8} & \heatmap{56} & \heatmap{38} & \heatmap{6}
& \heatmap{54} & \heatmap{44} & \heatmap{2} & \heatmap{54} & \heatmap{42} & \heatmap{4} \\

\MRC L1 
& \heatmap{50} & \heatmap{44} & \heatmap{6} & \heatmap{48} & \heatmap{48} & \heatmap{4}
& \heatmap{50} & \heatmap{42} & \heatmap{8} & \heatmap{42} & \heatmap{46} & \heatmap{12}
& \heatmap{54} & \heatmap{40} & \heatmap{6} & \heatmap{50} & \heatmap{46} & \heatmap{4} \\

\MRC LF 
& \heatmap{50} & \heatmap{50} & \heatmap{0} & \heatmap{53} & \heatmap{43} & \heatmap{4}
& \heatmap{42} & \heatmap{48} & \heatmap{10} & \heatmap{48} & \heatmap{48} & \heatmap{4}
& \heatmap{60} & \heatmap{40} & \heatmap{0} & \heatmap{67} & \heatmap{29} & \heatmap{4} \\
\hline
{\bf Diversity}
& Win & Loss & Tie & Win & Loss & Tie
& Win & Loss & Tie & Win & Loss & Tie
& Win & Loss & Tie & Win & Loss & Tie \\
\hline
\KHC L1   
& \heatmap{46} & \heatmap{50} & \heatmap{4} & \heatmap{52} & \heatmap{44} & \heatmap{4}
& \heatmap{46} & \heatmap{52} & \heatmap{2} & \heatmap{36} & \heatmap{58} & \heatmap{6}
& \heatmap{52} & \heatmap{48} & \heatmap{0} & \heatmap{48} & \heatmap{52} & \heatmap{0} \\

\KHC LF   
& \heatmap{56} & \heatmap{42} & \heatmap{2} & \heatmap{44} & \heatmap{52} & \heatmap{4}
& \heatmap{48} & \heatmap{44} & \heatmap{8} & \heatmap{56} & \heatmap{44} & \heatmap{0}
& \heatmap{46} & \heatmap{52} & \heatmap{2} & \heatmap{44} & \heatmap{50} & \heatmap{6} \\

\MHC L1  
& \heatmap{64} & \heatmap{32} & \heatmap{4} & \heatmap{56} & \heatmap{42} & \heatmap{2}
& \heatmap{44} & \heatmap{52} & \heatmap{4} & \heatmap{44} & \heatmap{50} & \heatmap{6}
& \heatmap{54} & \heatmap{46} & \heatmap{0} & \heatmap{50} & \heatmap{50} & \heatmap{0} \\

\MHC LF  
& \heatmap{66} & \heatmap{34} & \heatmap{0} & \heatmap{55} & \heatmap{44} & \heatmap{1}
& \heatmap{56} & \heatmap{42} & \heatmap{2} & \heatmap{61} & \heatmap{37} & \heatmap{2}
& \heatmap{52} & \heatmap{48} & \heatmap{0} & \heatmap{54} & \heatmap{36} & \heatmap{10} \\

\MRC L1 
& \heatmap{50} & \heatmap{50} & \heatmap{0} & \heatmap{46} & \heatmap{54} & \heatmap{0}
& \heatmap{62} & \heatmap{38} & \heatmap{0} & \heatmap{48} & \heatmap{44} & \heatmap{8}
& \heatmap{56} & \heatmap{40} & \heatmap{4} & \heatmap{54} & \heatmap{44} & \heatmap{2} \\

\MRC LF 
& \heatmap{52} & \heatmap{48} & \heatmap{0} & \heatmap{49} & \heatmap{51} & \heatmap{0}
& \heatmap{52} & \heatmap{48} & \heatmap{0} & \heatmap{56} & \heatmap{44} & \heatmap{0}
& \heatmap{62} & \heatmap{38} & \heatmap{0} & \heatmap{69} & \heatmap{28} & \heatmap{3} \\
\hline
\end{tabular}
}
\vspace{-2ex}
\end{table*}

\section{Results and Analysis} \label{sec:results_analysis}

Knowledge cutoff is a key consideration in RAG evaluation \cite{yang26, seo25}, as models may have already seen the corpus if it predates their training data. To minimize such contamination, we prioritize \textbf{post-cutoff data} when selecting both datasets and evaluation models. 

Our primary evaluation (Section~\ref{subsec:pri}) uses \threepointfive and \four
across three post-cutoff datasets: \podcast, \news, and \semi.
 This setup ensures a fair comparison of GraphRAG configurations under realistic retrieval and sensemaking conditions, while minimizing the influence of memorized knowledge.
To further verify that our heuristics generalize to more recent, stronger models, we conduct a secondary evaluation using \five on the full datasets, to account for the limited post-cutoff content (Section~\ref{subsec:sec}).
We then report the statistical significance of our $k$-core–based heuristics (Section~\ref{subsec:stats}). Finally, we analyze token efficiency compared to Edge et al.~\cite{edge24} and evaluate our \textit{Round-Robin Token-Constrained Selection} (RRTC) heuristic that aims to reduce token usage while maintaining competitive performance (Section~\ref{subsec:token}).

\subsection{Results on Post–Cutoff Data}\label{subsec:pri}

In this section, we evaluate the performance of our heuristics vs. Leiden-based GraphRAG~\cite{edge24} on three datasets under a post–knowledge-cutoff setting, using \threepointfive and \four, according to the two metrics defined earlier: {comprehensiveness} and {diversity}.
We also perform factual grounding analysis, human validation, sensitivity analysis for maximum cluster size, and runtime analysis. 

Table~\ref{tab:gpt35} summarizes the \threepointfive results.
We report head-to-head comparisons of our $k$-core based heuristics (in rows) vs. Leiden-based GraphRAG~\cite{edge24} (columns) on post-knowledge-cutoff content (September 2021), covering 183 documents for \semi, 609 for \news, and 36 for \podcast.
For each of our heuristics, we consider the leaf (LF) and parent-of-leaf (L1) levels, and for Leiden-based GraphRAG, we consider C2 (third from top) and C3 (fourth from top) levels, which consistently perform well (see Section~\ref{sec:back} for details).
 Results on C0 and C1 are given in the 
Appendix (Table \ref{tab:gpt3_c0_c1}).

Across all datasets and configurations, our $k$-core–based heuristics achieve higher win rates than the C2 and C3 configurations of \citeauthor{edge24} in approximately \textbf{70--75\%} of comparisons.
 Ties are generally infrequent (typically below \textbf{10\%}) and nearly absent for \semi, indicating decisive preferences.
{\bf \MHC LF is the most consistently strong configuration in the evaluation}. It never records a negative net win rate across any dataset, condition, or metric. On comprehensiveness, it maintains steady positive margins (+8 to +18 net) across all three datasets against both C2 and C3, with no sharp degradation between conditions.
Performance is even stronger for diversity: it achieves the highest \podcast win rate in the table (66\% against C2), delivers strong \news results that actually improve from C2 to C3 (reaching 61/37 in C3), and posts solid \semi gains especially in C3 (+18 net).
Dataset-wise, gains are largest on \semi, where \MRC LF averages \textbf{$\sim$64\%} wins and peaks at \textbf{67--69\%} against C3. 
For \podcast and \news, \MHC LF performs best, averaging \textbf{$\sim$57\%} and \textbf{$\sim$56\%}, respectively. 
Across datasets, our heuristics outperform C3 more strongly than C2 by roughly \textbf{3--6\%} on average.

Averaged across datasets and community levels, leaf-level (LF) variants consistently outperform their L1 counterparts by \textbf{5--10} percentage points, confirming that finer-grained communities lead to more informative and diverse summaries. 
This effect is particularly pronounced on \podcast and \semi, where LF configurations frequently exceed \textbf{55--60\%} win rates across both metrics.
Overall, \textbf{\MHC LF} and \textbf{\MRC LF} emerge as the strongest heuristics: \MHC LF shows the most consistent gains for diversity on \podcast and \news, while \MRC LF dominates comprehensiveness on \semi.

We conduct the same evaluation using \four, restricting all datasets to content published after its knowledge cutoff (October 2023). This results in 63 documents for \semi, 331 documents for \news, and 13 documents for  \podcast.
Overall, the trends remain consistent with \threepointfive, though performance margins are narrower and ties more frequent, as expected given that \four is a stronger model and the post-cutoff subsets are smaller.
The detailed results are given in the 
Appendix (Table \ref{tab:gpt4}).\\

\noindent{\bf Analysis of Factual Grounding.}
Our method maintains explicit entity/node-level grounding (similar to GraphRAG \cite{edge24}): (1) community reports store entity-linked information during indexing, and (2) generated answers include explicit references (e.g., [Data: Reports (255, 270, 287, 507, 171, +more)]) for each claim.

To evaluate grounding quality, we conducted an analysis on 25 randomly selected questions per dataset. We compute Claim Support Rate (CSR), defined as the fraction of answer claims supported by the referenced nodes, where an LLM verifies each claim against its cited nodes. Since some references are truncated via "+more'', we conservatively evaluate only the explicitly listed IDs. Under this setting, both $k$-core and Leiden achieve CSR $\approx 0.89$ on average, with the remaining cases attributable to evidence contained in the unexpanded "+more” references. These results indicate that both methods maintain strong and comparable factual grounding, while $k$-core achieves better performance on the primary global sensemaking tasks.\\

\noindent {\bf Human Validation of LLM Evaluations.}
No ground-truth dataset exists for global sensemaking tasks such as comprehensiveness and diversity, as these require aggregating information from hundreds of documents (up to 609 in our datasets), with answers ranging from 1,000–2,000 tokens, making manual verification nearly infeasible. Existing works, including the original GraphRAG \cite{edge24}, also rely on LLM judges for evaluation, while we strengthen our evaluation considering multiple judges/runs, randomize answer order to reduce bias, and apply post-cutoff filtering to mitigate memorization effects. Additionally, we conducted a small human evaluation on 25 randomly selected questions per dataset on the GPT-3.5 setup. The results show that human judgments agree with the LLM majority vote in 96\% of cases overall (Cohen’s $k = 0.94$), confirming that LLM-based comparisons reliably reflect answer quality.\\

\noindent {\bf Sensitivity Analysis of the Maximum Cluster Size.}
We performed a sensitivity analysis on the \podcast dataset to evaluate the effect of the maximum cluster size parameter, $M$, which is set heuristically as explained in Section~\ref{sec:proposed_heuristics} to ensure that important information is not truncated. $+5/+10$ change in $M$ reduces the win rate by $2$--$3\%$, likely due to context overflow, while $-5/-10$ change in $M$ increases the number of clusters and LLM calls without affecting overall performance. These results validate our choice of $M$.\\

\begin{table}[!t]
\centering
\caption{\five full: Head-to-head win rates (\%), for comprehensiveness and diversity metrics against Leiden C2 level from Edge et al.~\cite{edge24}. LF indicates leaf-level communities, and L1 indicates the level immediately above the leaf.}
\label{tab:gpt5_C2}
\resizebox{.45\textwidth}{!}{%
\begin{tabular}{|r|ccc|ccc|ccc|}
\hline
\five
& \multicolumn{3}{c|}{\podcast} 
& \multicolumn{3}{c|}{\news} 
& \multicolumn{3}{c|}{\ms} \\
\cline{2-10}
{\bf Comprehen.} 
& Win & Loss & Tie
& Win & Loss & Tie
& Win & Loss & Tie 
\\
\hline
\KHC L1   & \heatmap{41} & \heatmap{47} & \heatmap{12} & \heatmap{56} & \heatmap{42} & \heatmap{2} & \heatmap{43} & \heatmap{43} & \heatmap{14} \\
\KHC LF   & \heatmap{50} & \heatmap{48} & \heatmap{2}  & \heatmap{50} & \heatmap{45} & \heatmap{5} & \heatmap{45} & \heatmap{40} & \heatmap{15} \\
\MHC L1   & \heatmap{40} & \heatmap{52} & \heatmap{8}  & \heatmap{40} & \heatmap{34} & \heatmap{26} & \heatmap{40} & \heatmap{45} & \heatmap{15} \\
\MHC LF   & \heatmap{52} & \heatmap{42} & \heatmap{6}  & \heatmap{40} & \heatmap{40} & \heatmap{20} & \heatmap{48} & \heatmap{50} & \heatmap{2}  \\
\MRC L1   & \heatmap{48} & \heatmap{50} & \heatmap{2}  & \heatmap{42} & \heatmap{38} & \heatmap{20} & \heatmap{45} & \heatmap{48} & \heatmap{7}  \\
\MRC LF   & \heatmap{54} & \heatmap{44} & \heatmap{2}  & \heatmap{42} & \heatmap{36} & \heatmap{22} & \heatmap{45} & \heatmap{35} & \heatmap{20} \\
\hline
{\bf Diversity} 
& Win & Loss & Tie
& Win & Loss & Tie
& Win & Loss & Tie 
\\\hline
\KHC L1   & \heatmap{40} & \heatmap{44} & \heatmap{16} & \heatmap{50} & \heatmap{44} & \heatmap{6}  & \heatmap{48} & \heatmap{50} & \heatmap{2}  \\
\KHC LF   & \heatmap{50} & \heatmap{48} & \heatmap{2}  & \heatmap{45} & \heatmap{43} & \heatmap{12} & \heatmap{38} & \heatmap{45} & \heatmap{17} \\
\MHC L1   & \heatmap{37} & \heatmap{51} & \heatmap{12} & \heatmap{40} & \heatmap{38} & \heatmap{22} & \heatmap{45} & \heatmap{38} & \heatmap{17} \\
\MHC LF   & \heatmap{56} & \heatmap{36} & \heatmap{8}  & \heatmap{40} & \heatmap{37} & \heatmap{23} & \heatmap{52} & \heatmap{46} & \heatmap{2}  \\
\MRC L1   & \heatmap{48} & \heatmap{50} & \heatmap{2}  & \heatmap{36} & \heatmap{40} & \heatmap{24} & \heatmap{55} & \heatmap{42} & \heatmap{3}  \\
\MRC LF   & \heatmap{54} & \heatmap{44} & \heatmap{2}  & \heatmap{37} & \heatmap{37} & \heatmap{26} & \heatmap{40} & \heatmap{43} & \heatmap{17} \\
\hline
\end{tabular}
}
\end{table}

\noindent {\bf Runtime Analysis.}
Graph construction is identical across methods; runtime differences arise from community/hierarchy construction and the number of LLM calls, with more communities increasing cost. We measured wall-clock indexing time across the three datasets and find that $k$-core reduces hierarchy construction overhead by approximately $31\%$ on average compared to Leiden, demonstrating its practical efficiency. Additionally, as each LLM call has a similar cost, total runtime scales with the number of communities. Since $k$-core yields approximately $30$--$35\%$ fewer LLM calls on average (see Table~\ref{tab:community_stats_gpt3.5_post_cutoff}), it further improves end-to-end efficiency.

\subsection{Evaluation on Full Data by \five}\label{subsec:sec}

In our secondary experiments, we use \five to examine whether the trends observed with earlier models persist with a stronger, more recent evaluator. 
Due to \five’s later training window, none of the datasets have sufficient post-cutoff content. 
Hence, we evaluate on the full \news, \podcast, and \ms (since \semi would create an excessively large corpus).

Table~\ref{tab:gpt5_C2} reports head-to-head results between Leiden C2 and our $k$-core heuristics;  C3 results are given in 
Appendix  (Table \ref{tab:gpt5_full_c3}).
While our heuristics continue to show improvements, the high proportion of ties suggests that prior knowledge reduces the discriminative power of head-to-head evaluation.
Across datasets and metrics, leaf-level variants generally outperform L1 by {2--6} percentage points, with \MHC LF showing the strongest gains for diversity on \podcast and \MRC LF performing best for comprehensiveness on \ms.
While {\bf \MHC LF still achieves the strongest overall performance, \MRC LF
edges out \MHC LF on \podcast and \ms}.
Overall win rates remain around \textbf{45--55\%}, reflecting narrower margins
than with \threepointfive and slightly wider than with \four.
Several L1 configurations fall below 50\%, indicating that gains
are concentrated in leaf-level variants. Nevertheless, the directional consistency
with earlier evaluators suggests that $k$-core-based heuristics retain an
advantage, albeit a modest one, even under a stronger model whose prior knowledge
reduces the discriminative power of head-to-head evaluation.

\subsection{Statistical Analysis}\label{subsec:stats}

To assess statistical significance, we follow the procedure of \cite{edge24} and report $p$-values for pairwise comparisons of each heuristic against C2/C3 baselines on the comprehensiveness metric using \threepointfive post-cutoff data in Table~\ref{tab:p_values_summary} (details and diversity results are in Appendix~\ref{sec:p_values_appendix}
). Values with $p<0.005$ are highlighted in bold. Overall, \KHC and \MHC are significant on \texttt{podcast} and \texttt{news}, while \MHC and \MRC perform best on \texttt{semiconductor}. Notably, {\MHC} LF is the only heuristic significant ($p<0.005$) against both C2 and C3 across all datasets.\\

\begin{table}[!t]
\centering
\caption{Wilcoxon signed-rank test p-values for comprehensiveness (\threepointfive) across datasets and baselines. Values with $p<0.005$ are highlighted in \textbf{bold}.}
\label{tab:p_values_summary}
\resizebox{.46\textwidth}{!}{%
\begin{tabular}{p{.8cm} c c | c c | c c}
\hline
\textbf{Cond.} & \multicolumn{2}{c|}{\podcast} & \multicolumn{2}{c|}{\news} & \multicolumn{2}{c}{\semi} \\
 & C2 & C3 & C2 & C3 & C2 & C3 \\
\hline
\KHC L1   & \textbf{<0.001} & \textbf{<0.001} & 1.000  & \textbf{<0.001} & 1.000  & \textbf{<0.001} \\
\KHC LF   & \textbf{<0.001} & \textbf{<0.001} & \textbf{<0.001} & 0.370  & 0.982  & 0.363  \\
\MHC L1   & \textbf{<0.001} & \textbf{<0.001} & \textbf{<0.001} & \textbf{<0.001} & \textbf{<0.001} & 0.006  \\
\MHC LF   & \textbf{<0.001} & \textbf{<0.001} & \textbf{<0.001} & \textbf{<0.001} & \textbf{<0.001} & \textbf{<0.001} \\
\MRC L1   &  \textbf{0.007}  & 1.000  & \textbf{<0.001} & 0.069  & \textbf{<0.001} & 0.073  \\
\MRC LF   & 0.982  & \textbf{<0.001} & 0.007  & 1.000  & \textbf{<0.001} & \textbf{<0.001} \\
\hline
\end{tabular}%
}
\end{table}

\noindent {\bf Summary.} The results above show that the performance gap between our heuristics and Leiden-based C2/C3 is systematic rather than incidental. Knowledge graphs exhibit extremely sparse graph structures; e.g., average degrees in our datasets range from 2.88 to 4.42 and 55-60\% of nodes have a degree of just 1. This high prevalence of degree-1 nodes means that the majority of nodes have only a single connection, providing insufficient local structure for Leiden's modularity optimization to form meaningful communities, as the algorithm relies on dense local neighborhoods to effectively aggregate nodes into clusters.
As Theorem~\ref{thm:main} shows, Leiden’s degeneracy on such sparse graphs makes it unlikely to find the ``right'' partition, as there are exponentially many near-equivalent solutions, most of which group peripheral nodes incoherently.
In contrast, $k$-core–based heuristics produce a deterministic hierarchy that is
structurally grounded, explaining the consistent advantage of our methods across
\threepointfive, \four, and \five evaluations.
Note that nearly 45\% of clusters from \KHC have size 2, so \MHC and \MRC help improve performance by merging these, reducing fragmentation. Clusters of size 3 or 4 account for only 10–12\%, and merging them offers little benefit while inflating neighboring clusters, so they were excluded from our small-cluster definition.
Under \threepointfive, these differences are
statistically significant across all three datasets ($p < 0.005$ for \MHC LF),
and directionally consistent under \four and \five, where stronger prior knowledge
narrows the margins but does not eliminate the advantage.

\begin{table}[!t]
  \centering
  \caption{Community statistics across datasets, showing the number of communities and the percentage of source text used at different hierarchy levels and methods.}
  \label{tab:community_stats_gpt3.5_post_cutoff}
  \resizebox{.47\textwidth}{!}{%
  \begin{tabular}{lrrrrrrr}
  \hline
    Dataset & \# Tokens & Metric & C2 & C3 & \KHC LF & \MHC LF &  \MRC LF  \\
   \hline
\podcast & 180,267 & \# Comm. 
    & 291 & 301 & 351 & 311 & 194 \\
   &  & Coverage (\%)
    & 68.93 & 70.87 & 83.70 & 75.47 & 58.08 \\
  \hline
  \news & 1,340,078 & \# Comm.
    & 2,019 & 2,021 & 2,205 & 1,863 & 1,357 \\
   &  & Coverage (\%)
    & 65.90 & 71.26 & 76.46 & 65.74 & 55.08 \\
  \hline
  {\tt semic.} & 438,819 & \# Comm.
    & 627 & 696 & 847 & 705 & 468 \\
   &  & Coverage (\%)
    & 66.81 & 72.67 & 90.18 & 76.45 & 59.94 \\

   \hline
  \end{tabular}
  }
\end{table}

\subsection{Token Utilization and Impact of RRTC} \label{subsec:token}

Finally, we analyze how the evaluated methods consume tokens, and how our round-robin token-constrained selection (RRTC) mechanism (see Section~\ref{subsec:alg-token}) performs.
We first analyze the number of communities produced by each heuristic, along with the percentage of source tokens contained in those communities. These are the communities and tokens that are used to answer queries via LLM calls: fewer communities reduce the number of LLM calls, and lower token counts reduce overall LLM token usage.
Results on \threepointfive post-cutoff data are shown in Table~\ref{tab:community_stats_gpt3.5_post_cutoff} for leaf-level, which typically contain more communities and use more tokens than L1.
Notably, \MRC reduces both the number of communities and the tokens used: it consistently produces the fewest communities and lowest coverage (55–60\%), meaning it aggressively consolidates information.
 These reductions indicate that our heuristics use fewer tokens and require fewer LLM calls, while achieving comparable or improved performance on global sensemaking, as shown in the previous section.\\

\noindent {\bf Impact of Round-Robin Token-Constrained Selection (RRTC).} 
Table~\ref{tab:rr_l2tc_h2h} reports results for \MHC LF with RRTC across three datasets, compared against Leiden C2/C3 under three edge budgets: 80\%, 70\%, and 60\% of the original graph. 
The table reports token usage relative to C2/C3 and head-to-head win rates against them for comprehensiveness and diversity, where win rates are computed against losses only (ties excluded); thus, a 55\% win rate implies a 45\% loss rate.

Across datasets, RRTC maintains competitive performance even when selecting only 60\% of graph edges. Comprehensiveness win rates remain in the 50--56\% range, with the highest values on \semi at the 80\% budget, reaching 56\% and 58\% against C2 and C3, respectively.
Diversity is similarly preserved, often exceeding 55\% on \podcast and \semi despite substantial edge reduction. 
As the edge budget decreases from 80\% to 60\%, win rates gradually decline, with comprehensiveness on \podcast dropping below 50\%, indicating that a 60\%
budget may be too aggressive for some dataset--metric combinations. 
Overall, RRTC reduces token usage by up to $\sim$\textbf{40\%} relative to Leiden C2/C3 while maintaining competitive comprehensiveness and diversity, providing a practical, token-efficient alternative under moderate budget constraints.

\begin{table}[h]
\centering
\small
\caption{RRTC (\MHC Leaf) vs. Leiden C2/C3: Head-to-head win rates (\%) on comprehensiveness and diversity}
\label{tab:rr_l2tc_h2h}
\resizebox{.46\textwidth}{!}{%
\begin{tabular}{p{.8cm}c cc|cc | cc}
\hline
\textbf{Dataset} & \textbf{Edge} & \multicolumn{2}{c|}{\textbf{Relative token}} & \multicolumn{2}{c|}{\textbf{Comp.}} & \multicolumn{2}{c}{\textbf{Div.}} \\
& \textbf{budget} & \multicolumn{2}{c|}{\textbf{usage (\%)}} & \multicolumn{2}{c|}{\textbf{win rate (\%)}} & \multicolumn{2}{c}{\textbf{win rate (\%)}} \\

\cline{3-8}
                 &     (\%)       & C2 & C3 & C2 & C3 & C2 & C3 \\
\hline
\multirow{3}{*}{\podcast} 
 & 80 & 90 & 88 & 57 & 55 & 64 & 55 \\
 & 70 & 79 & 77 & 54 & 52 & 56 & 52 \\
 & 60 & 68 & 66 & 47 & 46 & 51 & 49 \\
\hline
\multirow{3}{*}{\news} 
 & 80 & 81 & 75 & 54 & 54 & 55 & 58 \\
 & 70 & 71 & 65 & 49 & 54 & 52 & 56 \\
 & 60 & 60 & 56 & 52 & 50 & 46 & 50 \\
\hline
\multirow{3}{*}{\tt semic.}
 & 80 & 92 & 85 & 56 & 58 & 55 & 60 \\
 & 70 & 81 & 74 & 53 & 53 & 50 & 57 \\
 & 60 & 69 & 64 & 55 & 50 & 52 & 54 \\
\hline
\end{tabular}
}
\end{table}

\section{Conclusion}\label{sec:conc}
We proved that modularity optimization is inherently unreliable on the sparse graphs typical of GraphRAG knowledge graphs, and proposed $k$-core–based heuristics that yield deterministic, size-bounded hierarchies in linear time. Across three datasets, three generator LLMs, and five LLM judges, our approach consistently improves comprehensiveness and diversity while reducing token usage compared to Leiden-based GraphRAG. Future work includes extending the framework to dynamic knowledge graphs and exploring its applicability beyond global sensemaking tasks.
In addition, alternative edge-selection mechanisms for token-budgeted retrieval and summarization, including PageRank-based methods and approaches inspired by \cite{hossain23}, remain an interesting direction for future research.

\section*{Acknowledgments}

Hossain and Sarıyüce were supported by NSF-2107089 and NSF-2236789 awards. This work used resources from the Center for Computational Research at the University at Buffalo~\cite{ccr} and the SUNY AI Platform on Google Cloud~\cite{sunyai}.

\bibliographystyle{ACM-Reference-Format}
\bibliography{reference}

%
%

\appendix

\section{Proof of Theorem~\ref{thm:main} and Experimental Validation}\label{app:proof}

\noindent {\bf Proof sketch.}

\noindent\textbf{Step 1: Low-degree nodes are weakly coupled.}
A node with degree $k_i \leq d$ participates in at most $d$ edges. Moving it between communities changes modularity by at most $O(d/m)$: the adjacency term shifts by at most $d/m$, and the null-model penalty shifts by $O(k_i \cdot K_r/(2m)^2)$ where $K_r$ is the total degree of the destination community. In the sparse regime where $m = \Theta(n)$, this is $O(1/n)$, so the node's community assignment barely affects the objective.

\medskip
\noindent\textbf{Step 2: Weakly coupled nodes can be independently reassigned.}
Reassigning one low-degree node perturbs another's sensitivity by only $O(1/n^2)$. By selecting a large independent set among the low-degree nodes, which is straightforward since these nodes have bounded degree, we obtain $\Theta(n)$ nodes whose reassignments are approximately independent. Each can be placed in at least two communities without meaningful modularity loss, yielding $2^{\Theta(n)}$ distinct near-optimal partitions.

\medskip
\noindent\textbf{Step 3: Sparsity makes this severe.}
When $\kbar < 3$, standard results on Poisson and power-law degree distributions guarantee that $n_{\mathrm{\le d}} = \Theta(n)$. This is the regime of knowledge graphs in GraphRAG. In contrast, for dense graphs with large $\kbar$, $n_{\mathrm{\le d}}/n \to 0$ and most nodes are strongly coupled to their communities, collapsing the degeneracy.\\

\noindent {\bf Full proof.}

\begin{proof}
We first remind modularity in its standard community-sum form:
\[
Q(\sigma) = \sum_{c} \left[\frac{e_c}{m} - \left(\frac{K_c}{2m}\right)^{\!2}\right],
\]
where $e_c$ counts edges internal to community $c$ and $K_c = \sum_{j \in c} k_j$ is its total degree. Moving node $i$ from community $s$ to community $r$ produces a new partition $\sigma^{i \to r}$. We write $d_i^{(c)}$ for the number of neighbors of $i$ in community $c$.
We also define the \emph{modularity sensitivity} of node $i$ under partition $\sigma$ is $\Delta_i(\sigma) = \max_{r \neq \sigma_i} |Q(\sigma) - Q(\sigma^{i \to r})|$, where $\sigma^{i \to r}$ moves $i$ to community $r$ with all other assignments fixed. Node $i$ is \emph{$\varepsilon$-weakly coupled} if $\Delta_i(\sigma) < \varepsilon$.

Next, we give two lemmas that will help with the proof

\begin{lemma}[Moving a low-degree node]\label{lem:one}
For any partition $\sigma$ and any node $i$ with degree $k_i \leq d$, moving $i$ to any other community changes modularity by at most
\[
\Delta_i(\sigma) \;\leq\; \frac{2k_i}{m} + \frac{k_i^2}{2m^2} \;=\; O\!\left(\frac{1}{n}\right).
\]
\end{lemma}

\begin{proof}
Moving $i$ from community $s$ to $r$ affects only those two terms in the sum. There are two contributions:

\medskip
\noindent\textit{Edge term.} Community $s$ loses $d_i^{(s)}$ internal edges; $r$ gains $d_i^{(r)}$. Both are at most $k_i$, so
\[
|\Delta_{\mathrm{edge}}| = \frac{|d_i^{(r)} - d_i^{(s)}|}{m} \leq \frac{k_i}{m}.
\]

\noindent\textit{Degree-penalty term.} The total degree of $s$ drops by $k_i$ and that of $r$ rises by $k_i$. Expanding the squares and simplifying:
\[
|\Delta_{\mathrm{penalty}}| = \frac{|k_i[2(K_r - K_s) + 2k_i]|}{(2m)^2} \leq \frac{k_i}{m} + \frac{k_i^2}{2m^2},
\]
using $|K_r - K_s| \leq 2m$.

\medskip
\noindent Adding the two and substituting $k_i \leq d$ and $m = n\kbar/2 = \Theta(n)$ gives $\Delta_i \leq 2d/m + d^2/(2m^2)$. Both terms are $O(1/n)$.
\end{proof}

\begin{lemma}[Non-adjacent low-degree nodes]\label{lem:two}
Let $i,j$ be non-adjacent nodes with degrees $\leq d$. Reassigning $j$ changes node $i$'s sensitivity by at most
\[
|\Delta_i(\sigma) - \Delta_i(\sigma')| \;\leq\; \frac{2d^2}{(2m)^2} \;=\; O\!\left(\frac{1}{n^2}\right).
\]
\end{lemma}

\begin{proof}
Since $i$ and $j$ share no edge, reassigning $j$ does not change $i$'s neighbor counts $d_i^{(c)}$ in any community, so the edge term for $i$ is unaffected. The only effect is that $j$'s move shifts the total degrees $K_r, K_s$ by $\pm k_j \leq d$, which perturbs $i$'s penalty term by at most $k_i \cdot 2d/(2m)^2 \leq 2d^2/(2m)^2 = O(1/n^2)$.
\end{proof}

\noindent Now, we give the proof of Theorem~\ref{thm:main}. The argument has three parts.

\medskip
\noindent\textbf{1. Find a large independent set of weakly coupled nodes.}
By Lemma~\ref{lem:one}, every node with $k_i \leq d$ has sensitivity $O(1/n)$, so for $\varepsilon > d(2+\kbar)/(2m)$ it is $\varepsilon$-weakly coupled. Let $n_{\le d}$ denote the number of such nodes. The subgraph they induce has maximum degree $d$, so a greedy algorithm yields an independent set $S$ of size
\[
|S| \;\geq\; \frac{n_{\le d}}{d+1}.
\]

\medskip
\noindent\textbf{2. Every combination of reassignments stays near-optimal.}
Start from an optimal partition $\sigma^*$. Each node $i \in S$ can stay or move to an alternative community: two choices per node. For any subset $T \subseteq S$ of nodes that are reassigned, the total modularity loss is bounded by summing the individual sensitivities plus pairwise cross-terms (using Lemma~\ref{lem:two}):
\[
|Q^* - Q(\sigma_T)| \;\leq\; \underbrace{\sum_{i \in T} \Delta_i(\sigma^*)}_{\text{individual}} \;+\; \underbrace{\sum_{\{i,j\} \subseteq T} \frac{2d^2}{(2m)^2}}_{\text{cross-terms}}.
\]
Bounding each sum with $|T| \leq |S| \leq n_{\le d}/(d+1)$ and $m = n\kbar/2$:
\begin{align*}
\text{Individual sum} &\leq \frac{n_{\le d}}{d+1} \cdot \frac{d(2+\kbar)}{n\kbar} \leq \frac{d(2+\kbar)}{(d+1)\kbar} \eqqcolon C_1, \\[6pt]
\text{Cross-term sum} &\leq \frac{1}{2}\left(\frac{n_{\le d}}{d+1}\right)^{2} \cdot \frac{2d^2}{n^2\kbar^2} \leq \frac{d^2}{(d+1)^2\kbar^2} \eqqcolon C_2.
\end{align*}
Both $C_1$ and $C_2$ are constants (independent of $n$), so setting $\varepsilon' = C_1 + C_2$ keeps every such partition within $\varepsilon'$ of $Q^*$.

\medskip
\noindent\textbf{3. Count and conclude.}
Each of the $|S|$ nodes has two valid assignments, giving $2^{|S|}$ structurally distinct near-optimal partitions:
\[
\mathcal{D}(\varepsilon') \;\geq\; 2^{|S|} \;\geq\; 2^{\,n_{\le d}/(d+1)}.
\]
In the sparse regime ($\kbar < 3$), standard degree distributions guarantee $n_{\le d} = \Theta(n)$. For example, under a Poisson distribution with $\kbar = 2$ and $d = 2$, the fraction of low-degree nodes is $5e^{-2} \approx 0.68$; under a power law with exponent $\gamma = 2.5$, degree-1 nodes alone comprise ${\sim}74\%$ of the graph. Therefore $\mathcal{D}(\varepsilon') \geq 2^{\Theta(n)}$.
\end{proof}

\subsection{Empirical Analysis of Leiden Instability}
To complement Theorem~\ref{thm:main}, we conducted a multi-seed stability analysis (10 runs with different random seeds) of Leiden across 3 datasets (GPT-3.5 setup). 
Across runs, the cluster statistics show noticeable variability. For example, on the Podcast dataset, the number of clusters ranges from 2550 to 2584, while the largest cluster size ranges between 780 and 1069 nodes. The average node-level Adjusted Rand Index (ARI) is 0.94, indicating around 6\% disagreement in node-pair assignments.
While moderate, this level of variation is significant in GraphRAG as community assignments directly define retrieval units and contexts for the LLM: the Jaccard similarity of retrieved nodes across seeds is ~0.73, on average.
These findings provide direct empirical evidence of retrieval-level instability in Leiden clustering and further support Theorem~\ref{thm:main}.

\section{Splitting Oversized Components}

\subsection{SPLIT: Splitting Large Connected Components} \label{sec:split_cc}

When a connected component $R$ exceeds the maximum cluster size $M$, it is split into smaller clusters using the \textsc{Split} procedure. The algorithm (Algorithm~\ref{alg:split_cc}) grows each cluster greedily from a seed node with the highest degree (line~\ref{line:pick_seed}), adding neighboring nodes that maximize internal connectivity (line~\ref{line:add_best_node}) until the cluster reaches size $M$. Each completed cluster is added to the cluster set (line~\ref{line:update_cluster}) and removed from the remaining node set (line~\ref{line:remove_set}). The output is a set of size-constrained clusters that preserve internal connectivity within $C$.

\begin{algorithm}[!h]
\caption{\textsc{Split} (G, R, M)}
\label{alg:split_cc}
\begin{algorithmic}[1]
\REQUIRE Graph $G=(V,E)$, connected comp. $R$, max size $M$
\ENSURE List of split clusters $\mathcal{C}$

\STATE Initialize cluster set $\mathcal{C} \gets \emptyset$
\WHILE{$R \neq \emptyset$}
    \STATE Pick seed node $s \in R$ with highest degree \label{line:pick_seed}
    \STATE $S \gets \{s\}$, $frontier \gets N(s) \cap R$
    
    \WHILE{$|S| < M$ \textbf{and} $frontier \neq \emptyset$}
        \STATE Pick $v \in frontier$ maximizing $N(v) \cap S$ \label{line:add_best_node}
        \STATE $S \gets S \cup v$
        \STATE $frontier \gets N(v) \cap (R \setminus S)$
    \ENDWHILE
        \STATE $\mathcal{C} \gets \mathcal{C} \cup S $  \label{line:update_cluster}
    \STATE Remove $S$ from $R$ \label{line:remove_set}
\ENDWHILE

\RETURN $\mathcal{C}$
\end{algorithmic}
\end{algorithm}

\FloatBarrier

\subsection{SPLIT-2HOP: Splitting Large 2-Hop Connected Components}
\label{sec:split2hop}

When a 2-hop connected component $H$ exceeds the maximum cluster size $M$, it is partitioned into smaller clusters using the \textsc{Split-2HOP} procedure (Algorithm~\ref{alg:split2hop}). The algorithm first computes the anchor set as the union of neighbors of all 2-hop nodes and precomputes anchor neighborhoods for each node (lines~\ref{line:anchors}--\ref{line:Au}).  
Clusters are grown greedily by selecting a seed node with the largest anchor set (line~\ref{line:seed}) and initializing a cluster with this seed (line~\ref{line:initS}). Candidate nodes that share at least one anchor with the seed are added to the frontier (line~\ref{line:frontier}). The cluster is expanded by repeatedly selecting the frontier node with the highest overlap in anchor connectivity with the current cluster (lines~\ref{line:innerwhile}--\ref{line:addv}), until the cluster reaches size $M$ or no eligible nodes remain.  
After cluster growth, anchor nodes connected to at least two nodes in the cluster are selected and included (line~\ref{line:selectanchors}), and the resulting cluster is added to the output set (line~\ref{line:addcluster}). The assigned nodes are then removed from the remaining 2-hop groups set (line~\ref{line:removeS}). This process repeats until all nodes in the 2-hop groups are assigned.

\begin{algorithm}[!htb]
\caption{\textsc{Split-2hop} (G, H, M)}
\label{alg:split2hop}
\begin{algorithmic}[1]
\REQUIRE Graph $G=(V,E)$, 2-hop groups $H$, max size $M$
\ENSURE List of split clusters $\mathcal{C}$

\STATE Initialize cluster set $\mathcal{C} \gets \emptyset$ \label{line:initC}

\STATE Compute anchors $A \gets \bigcup_{u \in H} N(u)$ \label{line:anchors}
\STATE Precompute $A_u \gets N(u) \cap A$ for $u \in H$ \label{line:Au}

\WHILE{$H \neq \emptyset$} \label{line:outerwhile}
    \STATE Pick $seed \in H$ with max $|A_{seed}|$ \label{line:seed}
    \STATE $S \gets \{seed\}$ \label{line:initS}
    \STATE $frontier \gets \{v \in H\setminus\{seed\} \mid A_v \cap A_{seed} \neq \emptyset\}$ \label{line:frontier}

    \WHILE{$|S| < M$ and $frontier \neq \emptyset$} \label{line:innerwhile}
        \STATE Pick $v \in frontier$ maximizing $\sum_{u \in S} |A_v \cap A_u|$ \label{line:pickv}
        \STATE $S \gets S \cup \{v\}$\;
        $frontier \gets (frontier \cup \{u \in H \mid A_u \cap A_v \neq \emptyset\}) \setminus S$ \label{line:addv}
    \ENDWHILE
    
\tcp{Select anchors linked to $\ge 2$ nodes in $S$} 
    \STATE $A' \gets \{v \in A s.t. N(v) \cap S \ge 2\}$ \label{line:selectanchors}
    \STATE $\mathcal{C} \gets \mathcal{C} \cup \{S \cup A'\}$ \label{line:addcluster}
    \STATE $H \gets H \setminus S$ \label{line:removeS}
\ENDWHILE

\RETURN $\mathcal{C}$
\end{algorithmic}
\end{algorithm}

\section{Alternative Community Partition Methods} \label{sec:altcomdet}
To our knowledge, prior GraphRAG studies do not evaluate alternative community detection methods, making Leiden the de facto comparison baseline. We also explored hierarchical block structures \cite{peixoto14} and non-modularity-based approaches such as $k$-truss, but found them consistently weaker than Leiden and therefore omitted them from the experimental results (additional discussion on $k$-truss is provided in Section~\ref{sec:proposed_heuristics}). For example, on the Semiconductor dataset (GPT-3.5 setup), Leiden outperforms $k$-truss in $\sim$82--85\% of comparisons, and $k$-core outperforms $k$-truss in $\sim$90--93\% of comparisons.

\section{Evaluation Criteria Details} 
\subsection{Question Generation} \label{sec:question_generation}
Following methodology described in \cite{edge24}, we generated 125 sensemaking questions using \five.
Consistent with their approach, we first prompted \five to create personas of hypothetical users for each corpus and then generated 5 tasks per user.  
For each user-task pair, we use \five to generate high-level questions that require understanding of the entire corpus without relying on low-level fact retrieval.  
This process ensures that the resulting questions assess comprehensive, corpus-wide reasoning, aligned with prior work.

\subsection{Evaluation Approach} \label{sec:evaluation_details}
We evaluate models using a head-to-head framework, where LLMs compare pairs of answers to identify a winner, loser, or tie. To enhance reliability, multiple evaluators are used, and repeated assessments are performed. The final outcome is determined via majority voting, making this approach suitable for global sensemaking tasks without gold-standard references.
 Each comparison proceeds as follows:
\begin{enumerate}
\item Generate answers to the same query using two approaches.
\item Randomize answer order and present the pair to the evaluator.
\item Each of five independent LLM evaluators (\five, Gemini 3 Pro Preview, Gemini 2.5 Pro, Qwen3 Next 80B, and DeepSeek v3.2) assesses the same answer pair three times.
\item Apply majority voting within each evaluator to determine their decision.
\item Apply a second majority vote across evaluators to determine the final judgment.
\end{enumerate}
This procedure ensures robust, repeatable head-to-head comparisons across methods.


\section{Statistical Analysis (p-values)} \label{sec:p_values_appendix}
To assess the significance of observed differences, we follow a procedure similar to Edge et al.~\cite{edge24}. 
For each dataset, evaluator, and metric, we assign scores for each head-to-head comparison: the winning method receives a score of 100, the losing method receives 0, and in the event of a tie, both methods receive 50.  These scores are then averaged over all evaluators, repeated runs, and questions, as detailed in Section~\ref{sec:evaluation_method}.
Following~\cite{edge24}, we use non-parametric Wilcoxon signed-rank tests to evaluate pairwise performance differences between methods. 
Holm-Bonferroni adjustment is applied to correct for multiple comparisons.

\begin{table}[!t]
\centering
\caption{Wilcoxon signed-rank test p-values for comparisons across datasets and baselines on Diversity (\threepointfive). Significant values are highlighted in \textbf{bold} for $p<0.005$.}
\label{tab:p_values_diversity}
\resizebox{.47\textwidth}{!}{%
\begin{tabular}{l c c | c c | c c}
\hline
\textbf{Condition} & \multicolumn{2}{c|}{\podcast} & \multicolumn{2}{c|}{\news} & \multicolumn{2}{c}{\semi} \\
 & C2 & C3 & C2 & C3 & C2 & C3 \\
\hline
\KHC L1   
& 0.073 & \textbf{<0.001} 
& 0.008 & \textbf{<0.001} 
& 0.083 & 0.083 \\

\KHC LF   
& \textbf{<0.001} & \textbf{<0.001} 
& 0.071 & \textbf{<0.001} 
& 0.008 & 0.007 \\

\MHC L1   
& \textbf{<0.001} & \textbf{<0.001} 
& \textbf{<0.001} & 0.007 
& \textbf{<0.001} & 0.982 \\

\MHC LF   
& \textbf{<0.001} & \textbf{<0.001} 
& \textbf{<0.001} & \textbf{<0.001} 
& 0.083 & \textbf{<0.001} \\

\MRC L1   
& 0.982 & \textbf{<0.001} 
& \textbf{<0.001} & 0.071 
& \textbf{<0.001} & \textbf{<0.001} \\

\MRC LF   
& 0.083 & 0.393 
& 0.083 & \textbf{<0.001} 
& \textbf{<0.001} & \textbf{<0.001} \\

\hline
\end{tabular}%
}
\end{table}


\subsection{Head-to-Head Comparison on Leiden C2 and Different \MHC Levels}

Table~\ref{tab:gpt35_msft_mhc_h2h} shows the head-to-head win rates (\%) for \threepointfive on \ms data, comparing the Leiden C2 baseline against various \MHC level variants. The two subtables report metrics for (a) Comprehensiveness and (b) Diversity.  

Here, the \MHC variants correspond to different k-core levels within the graph hierarchy: LF represents the leaf-level communities, LF-2 corresponds to two levels above the leaf, LF-4 to four levels above the leaf, and LF-6 to six levels above the leaf. We observe that higher k-core levels consistently achieve better win rates against C2, particularly in Comprehensiveness, indicating that the more central nodes capture global context more effectively. 
This justifies focusing on leaf-level (LF) and the level immediately above leaf (L1) for all of our evaluations, as they provide better performance.

\begin{table}[h!]
\centering
\caption{\threepointfive Head-to-Head Win Rates (\%) using Leiden C2 and \MHC Level Variants on \ms data}
\label{tab:gpt35_msft_mhc_h2h}
\setlength{\tabcolsep}{2pt} 

\begin{subtable}[t]{\columnwidth}
\centering
\caption{Comprehensiveness}
\begin{tabular}{|l|c|c|c|c|c|}
\hline
\textbf{} & Leiden C2 & \MHC LF-6 & \MHC LF-4 & \MHC LF-2 & \MHC LF \\
\hline
Leiden C2 & \heatmap{50} & \heatmap{56} & \heatmap{51} & \heatmap{48} & \heatmap{39} \\
\MHC LF-6 & \heatmap{44} & \heatmap{50} & \heatmap{45} & \heatmap{43} & \heatmap{33} \\
\MHC LF-4 & \heatmap{49} & \heatmap{55} & \heatmap{50} & \heatmap{46} & \heatmap{40} \\
\MHC LF-2 & \heatmap{52} & \heatmap{57} & \heatmap{54} & \heatmap{50} & \heatmap{47} \\
\MHC LF   & \heatmap{61} & \heatmap{67} & \heatmap{60} & \heatmap{53} & \heatmap{50} \\
\hline
\end{tabular}
\end{subtable}

\vspace{0.5em} 

\begin{subtable}[t]{\columnwidth}
\centering
\caption{Diversity}
\begin{tabular}{|l|c|c|c|c|c|}
\hline
\textbf{} & Leiden C2 & \MHC LF-6 & \MHC LF-4 & \MHC LF-2 & \MHC LF \\
\hline
Leiden C2 & \heatmap{50} & \heatmap{62} & \heatmap{54} & \heatmap{46} & \heatmap{39} \\
\MHC LF-6 & \heatmap{38} & \heatmap{50} & \heatmap{42} & \heatmap{38} & \heatmap{36} \\
\MHC LF-4 & \heatmap{46} & \heatmap{58} & \heatmap{50} & \heatmap{49} & \heatmap{41} \\
\MHC LF-2 & \heatmap{54} & \heatmap{62} & \heatmap{51} & \heatmap{50} & \heatmap{51} \\
\MHC LF   & \heatmap{61} & \heatmap{64} & \heatmap{59} & \heatmap{49} & \heatmap{50} \\
\hline
\end{tabular}
\end{subtable}

\end{table}

\section{Additional Results}

Here, we present additional results from \threepointfive, \four, and \five, including the results from Leiden C0, C1, and other metrics such as Empowerment and Directness.

\subsection{\four Results on Post-cutoff Data}
Table~\ref{tab:gpt4} reports head-to-head win rates for comprehensiveness and diversity under the same experimental conditions as in \threepointfive.
Across datasets and configurations, our heuristics still achieve higher win rates than Leiden C2/C3 in a majority of comparisons, averaging approximately \textbf{45--55\%}, particularly on \semi.
Among our methods, {\bf \MHC LF still achieves the strongest overall performance}, with average win rates of approximately \textbf{50--52\%} across datasets and metrics, followed by \KHC LF and \MRC LF at \textbf{48--50\%}. 
Dataset-wise, gains are most pronounced on \semi, where \MHC LF reaches up to \textbf{64\%} wins for diversity (C2) and maintains \textbf{58\%} against C3, while comprehensiveness peaks at \textbf{52\%}. 
For \podcast and \news, win rates remain closer to parity, typically fluctuating between \textbf{45--50\%} across heuristics and Leiden levels.

Comparing Leiden resolutions, differences between C2 and C3 are smaller than for \threepointfive, with average win-rate gaps generally within \textbf{2--4\%}. 
Averaged across datasets and metrics, leaf-level (LF) variants continue to outperform L1, but by a narrower margin of \textbf{2--5} percentage points.
Overall, \four narrows the gap between $k$-core heuristics and Leiden GraphRAG, with \textbf{\MHC LF} remaining the most robust, especially on \semi, and generally matching or slightly surpassing Leiden baselines despite more ties.

Note that performance on the \podcast is less reliable due to the limited number of post-cutoff documents. With only 13 documents available, the resulting graphs are sparse, and our heuristics often show minimal improvement, performing similarly to C2/C3 configurations.

\begin{table}[!ht]
\centering
\caption{\five full: Head-to-head win rates (\%), for comprehensiveness and diversity metrics. C3 is the Leiden community level from Edge et al.~\cite{edge24}. LF indicates leaf-level communities, and L1 indicates the level immediately above the leaf.}
\label{tab:gpt5_full_c3}
\resizebox{.47\textwidth}{!}{%
\begin{tabular}{|r|ccc|ccc|ccc|}
\hline
\five
& \multicolumn{3}{c|}{\podcast} 
& \multicolumn{3}{c|}{\news} 
& \multicolumn{3}{c|}{\ms} \\
\cline{2-10}
results for
& \multicolumn{3}{c|}{C3}
& \multicolumn{3}{c|}{C3}
& \multicolumn{3}{c|}{C3} \\
\cline{2-10}
{\bf Comprehensiveness}
& Win & Loss & Tie
& Win & Loss & Tie
& Win & Loss & Tie \\
\hline
\KHC L1   
& \heatmap{41} & \heatmap{41} & \heatmap{18}
& \heatmap{52} & \heatmap{39} & \heatmap{9}
& \heatmap{38} & \heatmap{43} & \heatmap{19} \\

\KHC LF   
& \heatmap{43} & \heatmap{43} & \heatmap{14}
& \heatmap{47} & \heatmap{45} & \heatmap{8}
& \heatmap{42} & \heatmap{48} & \heatmap{10} \\

\MHC L1  
& \heatmap{48} & \heatmap{48} & \heatmap{4}
& \heatmap{49} & \heatmap{37} & \heatmap{14}
& \heatmap{48} & \heatmap{42} & \heatmap{10} \\

\MHC LF  
& \heatmap{50} & \heatmap{47} & \heatmap{3}
& \heatmap{47} & \heatmap{35} & \heatmap{18}
& \heatmap{47} & \heatmap{48} & \heatmap{5} \\

\MRC L1 
& \heatmap{45} & \heatmap{45} & \heatmap{10}
& \heatmap{32} & \heatmap{48} & \heatmap{20}
& \heatmap{45} & \heatmap{40} & \heatmap{15} \\

\MRC LF 
& \heatmap{50} & \heatmap{37} & \heatmap{13}
& \heatmap{36} & \heatmap{46} & \heatmap{18}
& \heatmap{45} & \heatmap{40} & \heatmap{15} \\
\hline

{\bf Diversity}
& Win & Loss & Tie
& Win & Loss & Tie
& Win & Loss & Tie \\
\hline
\KHC L1   
& \heatmap{44} & \heatmap{44} & \heatmap{12}
& \heatmap{47} & \heatmap{38} & \heatmap{15}
& \heatmap{33} & \heatmap{53} & \heatmap{14} \\

\KHC LF   
& \heatmap{48} & \heatmap{45} & \heatmap{7}
& \heatmap{37} & \heatmap{34} & \heatmap{29}
& \heatmap{48} & \heatmap{44} & \heatmap{8} \\

\MHC L1  
& \heatmap{40} & \heatmap{42} & \heatmap{18}
& \heatmap{42} & \heatmap{38} & \heatmap{20}
& \heatmap{45} & \heatmap{48} & \heatmap{7} \\

\MHC LF  
& \heatmap{48} & \heatmap{41} & \heatmap{11}
& \heatmap{32} & \heatmap{33} & \heatmap{35}
& \heatmap{52} & \heatmap{40} & \heatmap{8} \\

\MRC L1 
& \heatmap{47} & \heatmap{46} & \heatmap{7}
& \heatmap{45} & \heatmap{29} & \heatmap{24}
& \heatmap{53} & \heatmap{35} & \heatmap{12} \\

\MRC LF 
& \heatmap{56} & \heatmap{31} & \heatmap{13}
& \heatmap{34} & \heatmap{30} & \heatmap{36}
& \heatmap{54} & \heatmap{38} & \heatmap{8} \\
\hline

\end{tabular}
}
\end{table}

\begin{table*}[!htb]
\centering
\caption{\threepointfive post-cutoff: Head-to-head win rates (\%), for comprehensiveness and diversity metrics. C0 and C1 are Leiden community levels from Edge et al.~\cite{edge24}. LF indicates leaf-level communities, and L1 indicates the level immediately above the leaf.}
\label{tab:gpt3_c0_c1}
\resizebox{\textwidth}{!}{%
\begin{tabular}{|r|ccc|ccc|ccc|ccc|ccc|ccc|}
\hline
\threepointfive
& \multicolumn{6}{c|}{\podcast} 
& \multicolumn{6}{c|}{\news} 
& \multicolumn{6}{c|}{\semi} \\
\cline{2-19}
results for
& \multicolumn{3}{c|}{C0} & \multicolumn{3}{c|}{C1}
& \multicolumn{3}{c|}{C0} & \multicolumn{3}{c|}{C1}
& \multicolumn{3}{c|}{C0} & \multicolumn{3}{c|}{C1} \\
\cline{2-19}
{\bf Comprehensiveness}
& Win & Loss & Tie & Win & Loss & Tie
& Win & Loss & Tie & Win & Loss & Tie
& Win & Loss & Tie & Win & Loss & Tie \\
\hline
\KHC L1   
& \heatmap{50} & \heatmap{42} & \heatmap{8} & \heatmap{68} & \heatmap{24} & \heatmap{8}
& \heatmap{56} & \heatmap{32} & \heatmap{12} & \heatmap{48} & \heatmap{42} & \heatmap{10}
& \heatmap{66} & \heatmap{30} & \heatmap{4} & \heatmap{72} & \heatmap{26} & \heatmap{2} \\

\KHC LF   
& \heatmap{58} & \heatmap{42} & \heatmap{0} & \heatmap{52} & \heatmap{44} & \heatmap{4}
& \heatmap{40} & \heatmap{50} & \heatmap{10} & \heatmap{42} & \heatmap{44} & \heatmap{14}
& \heatmap{68} & \heatmap{32} & \heatmap{0} & \heatmap{52} & \heatmap{46} & \heatmap{2} \\

\MHC L1  
& \heatmap{48} & \heatmap{46} & \heatmap{6} & \heatmap{66} & \heatmap{28} & \heatmap{6}
& \heatmap{50} & \heatmap{40} & \heatmap{10} & \heatmap{50} & \heatmap{36} & \heatmap{14}
& \heatmap{54} & \heatmap{46} & \heatmap{0} & \heatmap{62} & \heatmap{28} & \heatmap{10} \\

\MHC LF
& \heatmap{58} & \heatmap{40} & \heatmap{2} & \heatmap{68} & \heatmap{28} & \heatmap{4}
& \heatmap{44} & \heatmap{48} & \heatmap{8} & \heatmap{44} & \heatmap{49} & \heatmap{7}
& \heatmap{60} & \heatmap{38} & \heatmap{2} & \heatmap{72} & \heatmap{24} & \heatmap{4} \\

\MRC L1 
& \heatmap{58} & \heatmap{36} & \heatmap{6} & \heatmap{48} & \heatmap{48} & \heatmap{4}
& \heatmap{50} & \heatmap{38} & \heatmap{12} & \heatmap{46} & \heatmap{44} & \heatmap{10}
& \heatmap{54} & \heatmap{40} & \heatmap{6} & \heatmap{64} & \heatmap{32} & \heatmap{4} \\

\MRC LF 
& \heatmap{64} & \heatmap{36} & \heatmap{0} & \heatmap{57} & \heatmap{39} & \heatmap{4}
& \heatmap{56} & \heatmap{34} & \heatmap{10} & \heatmap{61} & \heatmap{35} & \heatmap{4}
& \heatmap{48} & \heatmap{42} & \heatmap{10} & \heatmap{55} & \heatmap{43} & \heatmap{2} \\
\hline
{\bf Diversity}
& Win & Loss & Tie & Win & Loss & Tie
& Win & Loss & Tie & Win & Loss & Tie
& Win & Loss & Tie & Win & Loss & Tie \\
\hline
\KHC L1   
& \heatmap{68} & \heatmap{28} & \heatmap{4} & \heatmap{60} & \heatmap{36} & \heatmap{4}
& \heatmap{60} & \heatmap{38} & \heatmap{2} & \heatmap{56} & \heatmap{38} & \heatmap{6}
& \heatmap{62} & \heatmap{38} & \heatmap{0} & \heatmap{66} & \heatmap{34} & \heatmap{0} \\

\KHC LF   
& \heatmap{52} & \heatmap{46} & \heatmap{2} & \heatmap{56} & \heatmap{40} & \heatmap{4}
& \heatmap{48} & \heatmap{48} & \heatmap{4} & \heatmap{50} & \heatmap{44} & \heatmap{6}
& \heatmap{62} & \heatmap{36} & \heatmap{2} & \heatmap{60} & \heatmap{34} & \heatmap{6} \\

\MHC L1  
& \heatmap{42} & \heatmap{54} & \heatmap{4} & \heatmap{70} & \heatmap{28} & \heatmap{2}
& \heatmap{56} & \heatmap{40} & \heatmap{4} & \heatmap{60} & \heatmap{34} & \heatmap{6}
& \heatmap{64} & \heatmap{36} & \heatmap{0} & \heatmap{70} & \heatmap{30} & \heatmap{0} \\

\MHC LF
& \heatmap{52} & \heatmap{48} & \heatmap{0} & \heatmap{72} & \heatmap{27} & \heatmap{1}
& \heatmap{60} & \heatmap{38} & \heatmap{2} & \heatmap{67} & \heatmap{32} & \heatmap{1}
& \heatmap{62} & \heatmap{38} & \heatmap{0} & \heatmap{52} & \heatmap{38} & \heatmap{10} \\

\MRC L1 
& \heatmap{66} & \heatmap{34} & \heatmap{0} & \heatmap{60} & \heatmap{40} & \heatmap{0}
& \heatmap{50} & \heatmap{48} & \heatmap{2} & \heatmap{42} & \heatmap{52} & \heatmap{6}
& \heatmap{52} & \heatmap{44} & \heatmap{4} & \heatmap{52} & \heatmap{46} & \heatmap{2} \\

\MRC LF 
& \heatmap{60} & \heatmap{40} & \heatmap{0} & \heatmap{65} & \heatmap{35} & \heatmap{0}
& \heatmap{64} & \heatmap{36} & \heatmap{0} & \heatmap{64} & \heatmap{36} & \heatmap{0}
& \heatmap{60} & \heatmap{40} & \heatmap{0} & \heatmap{56} & \heatmap{44} & \heatmap{0} \\
\hline
\end{tabular}
}
\end{table*}

\begin{table*}[!h]
\centering
\caption{\four post-cutoff : Head-to-head win rates (\%), for comprehensiveness and diversity metrics. C2 and C3 are Leiden community levels from Edge et al.~\cite{edge24}. LF indicates leaf-level communities, and L1 indicates the level immediately above the leaf.
}
\label{tab:gpt4}
\resizebox{\textwidth}{!}{%
\begin{tabular}{|r|ccc|ccc|ccc|ccc|ccc|ccc|}
\hline
\four
& \multicolumn{6}{c|}{\podcast} 
& \multicolumn{6}{c|}{\news} 
& \multicolumn{6}{c|}{\semi} \\
\cline{2-19}
results for
& \multicolumn{3}{c|}{C2} & \multicolumn{3}{c|}{C3}
& \multicolumn{3}{c|}{C2} & \multicolumn{3}{c|}{C3}
& \multicolumn{3}{c|}{C2} & \multicolumn{3}{c|}{C3} \\
\cline{2-19}
{\bf Comprehensiveness}
& Win & Loss & Tie & Win & Loss & Tie
& Win & Loss & Tie & Win & Loss & Tie
& Win & Loss & Tie & Win & Loss & Tie \\
\hline
\KHC L1   
& \heatmap{50} & \heatmap{50} & \heatmap{0} & \heatmap{50} & \heatmap{50} & \heatmap{0}
& \heatmap{38} & \heatmap{42} & \heatmap{20} & \heatmap{40} & \heatmap{42} & \heatmap{18}
& \heatmap{46} & \heatmap{46} & \heatmap{8} & \heatmap{39} & \heatmap{36} & \heatmap{25} \\

\KHC LF   
& \heatmap{50} & \heatmap{48} & \heatmap{2} & \heatmap{46} & \heatmap{48} & \heatmap{6}
& \heatmap{40} & \heatmap{40} & \heatmap{20} & \heatmap{43} & \heatmap{40} & \heatmap{17}
& \heatmap{52} & \heatmap{41} & \heatmap{7} & \heatmap{49} & \heatmap{32} & \heatmap{19} \\

\MHC L1  
& \heatmap{48} & \heatmap{46} & \heatmap{6} & \heatmap{48} & \heatmap{44} & \heatmap{8}
& \heatmap{42} & \heatmap{42} & \heatmap{16} & \heatmap{44} & \heatmap{44} & \heatmap{12}
& \heatmap{44} & \heatmap{40} & \heatmap{16} & \heatmap{51} & \heatmap{30} & \heatmap{19} \\

\MHC LF  
& \heatmap{50} & \heatmap{48} & \heatmap{2} & \heatmap{50} & \heatmap{46} & \heatmap{4}
& \heatmap{44} & \heatmap{40} & \heatmap{16} & \heatmap{46} & \heatmap{42} & \heatmap{12}

& \heatmap{52} & \heatmap{38} & \heatmap{10} & \heatmap{50} & \heatmap{24} & \heatmap{26} \\

\MRC L1 
& \heatmap{48} & \heatmap{44} & \heatmap{8} & \heatmap{44} & \heatmap{42} & \heatmap{14}
& \heatmap{38} & \heatmap{36} & \heatmap{26} & \heatmap{38} & \heatmap{32} & \heatmap{30}
& \heatmap{36} & \heatmap{42} & \heatmap{22} & \heatmap{46} & \heatmap{44} & \heatmap{10} \\

\MRC LF 
& \heatmap{48} & \heatmap{46} & \heatmap{6} & \heatmap{46} & \heatmap{48} & \heatmap{6}
& \heatmap{40} & \heatmap{35} & \heatmap{25} & \heatmap{40} & \heatmap{30} & \heatmap{30}
& \heatmap{48} & \heatmap{42} & \heatmap{10} & \heatmap{44} & \heatmap{44} & \heatmap{12} \\
\hline
{\bf Diversity}
& Win & Loss & Tie & Win & Loss & Tie
& Win & Loss & Tie & Win & Loss & Tie
& Win & Loss & Tie & Win & Loss & Tie \\
\hline
\KHC L1   
& \heatmap{50} & \heatmap{48} & \heatmap{2} & \heatmap{46} & \heatmap{48} & \heatmap{6}
& \heatmap{47} & \heatmap{40} & \heatmap{13} & \heatmap{52} & \heatmap{39} & \heatmap{9}
& \heatmap{38} & \heatmap{54} & \heatmap{8} & \heatmap{45} & \heatmap{48} & \heatmap{7} \\

\KHC LF   
& \heatmap{50} & \heatmap{50} & \heatmap{0} & \heatmap{48} & \heatmap{50} & \heatmap{2}
& \heatmap{50} & \heatmap{38} & \heatmap{12} & \heatmap{55} & \heatmap{37} & \heatmap{8}
& \heatmap{47} & \heatmap{49} & \heatmap{4} & \heatmap{52} & \heatmap{36} & \heatmap{12} \\

\MHC L1  
& \heatmap{50} & \heatmap{50} & \heatmap{0} & \heatmap{50} & \heatmap{50} & \heatmap{0}
& \heatmap{42} & \heatmap{48} & \heatmap{10} & \heatmap{44} & \heatmap{46} & \heatmap{10}
& \heatmap{50} & \heatmap{38} & \heatmap{12} & \heatmap{49} & \heatmap{36} & \heatmap{15} \\

\MHC LF  
& \heatmap{49} & \heatmap{46} & \heatmap{5} & \heatmap{46} & \heatmap{48} & \heatmap{6}
& \heatmap{44} & \heatmap{46} & \heatmap{10} & \heatmap{46} & \heatmap{44} & \heatmap{10}
& \heatmap{64} & \heatmap{32} & \heatmap{4} & \heatmap{58} & \heatmap{34} & \heatmap{8} \\

\MRC L1 
& \heatmap{48} & \heatmap{48} & \heatmap{4} & \heatmap{48} & \heatmap{44} & \heatmap{8}
& \heatmap{43} & \heatmap{34} & \heatmap{23} & \heatmap{46} & \heatmap{38} & \heatmap{16}
& \heatmap{46} & \heatmap{46} & \heatmap{8} & \heatmap{45} & \heatmap{45} & \heatmap{10} \\

\MRC LF 
& \heatmap{49} & \heatmap{50} & \heatmap{1} & \heatmap{48} & \heatmap{52} & \heatmap{0}
& \heatmap{45} & \heatmap{32} & \heatmap{23} & \heatmap{48} & \heatmap{36} & \heatmap{16}
& \heatmap{54} & \heatmap{44} & \heatmap{2} & \heatmap{54} & \heatmap{45} & \heatmap{1} \\
\hline
\end{tabular}
}
\end{table*}

\begin{table*}[!h]
\centering
\caption{\threepointfive post-cutoff: Head-to-head win rates (\%), for empowerment and directness metrics. C2/C3 are Leiden community levels, LF is leaf-level communities, and L1 is the level immediately above the leaf}
\label{tab:gpt35_empowerment_directness}
\resizebox{\textwidth}{!}{%
\begin{tabular}{|r|ccc|ccc|ccc|ccc|ccc|ccc|}
\hline
\threepointfive
& \multicolumn{6}{c|}{\podcast} 
& \multicolumn{6}{c|}{\news} 
& \multicolumn{6}{c|}{\semi} \\
\cline{2-19}
results for
& \multicolumn{3}{c|}{C2 } & \multicolumn{3}{c|}{C3 }
& \multicolumn{3}{c|}{C2 } & \multicolumn{3}{c|}{C3 }
& \multicolumn{3}{c|}{C2 } & \multicolumn{3}{c|}{C3 } \\
\cline{2-19}
{\bf Comprehensiveness}
& Win & Loss & Tie & Win & Loss & Tie
& Win & Loss & Tie & Win & Loss & Tie
& Win & Loss & Tie & Win & Loss & Tie \\
\hline
\KHC L1   
& \heatmap{48} & \heatmap{50} & \heatmap{2} & \heatmap{52} & \heatmap{45} & \heatmap{3}
& \heatmap{46} & \heatmap{52} & \heatmap{2} & \heatmap{63} & \heatmap{37} & \heatmap{0}
& \heatmap{52} & \heatmap{48} & \heatmap{0} & \heatmap{47} & \heatmap{52} & \heatmap{1} \\

\KHC LF   
& \heatmap{52} & \heatmap{43} & \heatmap{5} & \heatmap{54} & \heatmap{43} & \heatmap{3}
& \heatmap{46} & \heatmap{50} & \heatmap{4} & \heatmap{57} & \heatmap{40} & \heatmap{3}
& \heatmap{50} & \heatmap{50} & \heatmap{0} & \heatmap{50} & \heatmap{50} & \heatmap{0} \\

\MHC L1  
& \heatmap{54} & \heatmap{44} & \heatmap{2} & \heatmap{59} & \heatmap{39} & \heatmap{2}
& \heatmap{63} & \heatmap{35} & \heatmap{2} & \heatmap{55} & \heatmap{45} & \heatmap{0}
& \heatmap{40} & \heatmap{59} & \heatmap{1} & \heatmap{53} & \heatmap{43} & \heatmap{4} \\

\MHC LF  
& \heatmap{50} & \heatmap{46} & \heatmap{4} & \heatmap{54} & \heatmap{46} & \heatmap{0}
& \heatmap{46} & \heatmap{51} & \heatmap{3} & \heatmap{46} & \heatmap{54} & \heatmap{0}
& \heatmap{61} & \heatmap{38} & \heatmap{1} & \heatmap{52} & \heatmap{44} & \heatmap{4} \\

\MRC L1 
& \heatmap{39} & \heatmap{54} & \heatmap{7} & \heatmap{40} & \heatmap{55} & \heatmap{5}
& \heatmap{38} & \heatmap{60} & \heatmap{2} & \heatmap{38} & \heatmap{60} & \heatmap{2}
& \heatmap{51} & \heatmap{48} & \heatmap{2} & \heatmap{57} & \heatmap{37} & \heatmap{6} \\

\MRC LF 
& \heatmap{36} & \heatmap{56} & \heatmap{8} & \heatmap{47} & \heatmap{50} & \heatmap{3}
& \heatmap{50} & \heatmap{40} & \heatmap{10} & \heatmap{65} & \heatmap{30} & \heatmap{5}
& \heatmap{53} & \heatmap{46} & \heatmap{1} & \heatmap{59} & \heatmap{36} & \heatmap{5} \\
\hline
{\bf Directness}
& Win & Loss & Tie & Win & Loss & Tie
& Win & Loss & Tie & Win & Loss & Tie
& Win & Loss & Tie & Win & Loss & Tie \\
\hline
\KHC L1   
& \heatmap{62} & \heatmap{36} & \heatmap{2} & \heatmap{40} & \heatmap{50} & \heatmap{10}
& \heatmap{37} & \heatmap{63} & \heatmap{0} & \heatmap{45} & \heatmap{45} & \heatmap{10}
& \heatmap{35} & \heatmap{65} & \heatmap{0} & \heatmap{56} & \heatmap{43} & \heatmap{1} \\

\KHC LF   
& \heatmap{42} & \heatmap{55} & \heatmap{3} & \heatmap{35} & \heatmap{55} & \heatmap{10}
& \heatmap{40} & \heatmap{57} & \heatmap{3} & \heatmap{48} & \heatmap{42} & \heatmap{10}
& \heatmap{35} & \heatmap{65} & \heatmap{0} & \heatmap{58} & \heatmap{41} & \heatmap{1} \\

\MHC L1  
& \heatmap{54} & \heatmap{44} & \heatmap{2} & \heatmap{64} & \heatmap{28} & \heatmap{8}
& \heatmap{56} & \heatmap{44} & \heatmap{0} & \heatmap{56} & \heatmap{31} & \heatmap{13}
& \heatmap{63} & \heatmap{36} & \heatmap{1} & \heatmap{53} & \heatmap{43} & \heatmap{4} \\

\MHC LF  
& \heatmap{52} & \heatmap{44} & \heatmap{4} & \heatmap{62} & \heatmap{31} & \heatmap{7}
& \heatmap{46} & \heatmap{54} & \heatmap{0} & \heatmap{58} & \heatmap{31} & \heatmap{11}
& \heatmap{61} & \heatmap{38} & \heatmap{1} & \heatmap{53} & \heatmap{44} & \heatmap{3} \\

\MRC L1 
& \heatmap{43} & \heatmap{50} & \heatmap{7} & \heatmap{50} & \heatmap{44} & \heatmap{6}
& \heatmap{48} & \heatmap{50} & \heatmap{2} & \heatmap{59} & \heatmap{28} & \heatmap{13}
& \heatmap{50} & \heatmap{50} & \heatmap{0} & \heatmap{47} & \heatmap{53} & \heatmap{0} \\

\MRC LF 
& \heatmap{58} & \heatmap{36} & \heatmap{6} & \heatmap{47} & \heatmap{50} & \heatmap{3}
& \heatmap{48} & \heatmap{52} & \heatmap{0} & \heatmap{61} & \heatmap{28} & \heatmap{11}
& \heatmap{50} & \heatmap{50} & \heatmap{0} & \heatmap{45} & \heatmap{55} & \heatmap{0} \\
\hline
\end{tabular}
}
\vspace{-10ex}
\end{table*}

\end{document}